\documentclass[%
 reprint,
amsmath,amssymb,
aps,prx,twocolumn, superscriptaddress, showpacs
]{revtex4-2}

\def \red #1 {\textcolor{red}{#1}}

\usepackage{graphicx}% Include figure files
\usepackage[utf8]{inputenc}
\usepackage[braket, qm]{qcircuit}
\usepackage[T1]{fontenc}
 \usepackage{verbatim}
\usepackage{float}
\usepackage{dcolumn} 
\usepackage{bm}
\usepackage{mathtools}
\usepackage{xcolor}
\usepackage{framed}
\usepackage{setspace} 
\usepackage{algorithm}
\usepackage{algpseudocode}
\usepackage{subfigure}
\usepackage{enumitem}
\usepackage{amsmath}
\usepackage{amsthm}
\usepackage{multirow}

\usepackage{hyperref}

\usepackage{algpseudocode}

\newcommand\Algphase[1]{%
\vspace*{-.7\baselineskip}\Statex\hspace*{\dimexpr-\algorithmicindent-2pt\relax}\rule{0.49\textwidth}{0.4pt}%
\Statex\hspace*{-\algorithmicindent}\textbf{#1}%
\vspace*{-.7\baselineskip}\Statex\hspace*{\dimexpr-\algorithmicindent-2pt\relax}\rule{0.49\textwidth}{0.4pt}%
}
\newtheorem{theorem}{Result}
\newtheorem{lemma}{Lemma}
\newtheorem{Decomposition}{Decomposition}
\hypersetup{
    colorlinks=true,
    linkcolor=blue,
    citecolor=blue,
    filecolor=blue,
    urlcolor=blue
}
 
\begin{document}

%\preprint{APS/123-QED}

\title{Direct Measurement of Density Matrices via Dense Dual Bases}% Force line breaks with \\

\author{Yu Wang}
%\email{wangyu@bimsa.cn}
\affiliation{Beijing Institute of Mathematical Sciences and Applications}

\author{Hanru Jiang}
%\email{hanru@bimsa.cn}
\affiliation{Beijing Institute of Mathematical Sciences and Applications}

\author{Yongxiang Liu}
\email{liuyx@pcl.ac.cn}
\affiliation{Peng Cheng Laboratory, Shenzhen 518055, China}

\author{Keren Li}
\email{likr@szu.edu.cn}
\affiliation{College of Physics and Optoelectronic Engineering, Shenzhen University, Shenzhen 518060, China}
\affiliation{Quantum Science Center of Guangdong-Hong Kong-Macao Greater Bay Area
(Guangdong), Shenzhen 518045. China}

\begin{abstract} 

 Efficient understanding of a quantum system fundamentally relies on the selection of observables. Pauli observables and mutually unbiased bases (MUBs) are widely used in practice and are often regarded as theoretically optimal for quantum state tomography (QST). However, Pauli observables require a large number of measurements for full-state tomography and do not permit direct measurement of density matrix elements with a constant number of observables. For MUBs, the existence of complete sets of \(d+1\) bases in all dimensions remains unresolved, highlighting the need for alternative observables. 
In this work, we introduce Dense Dual Bases (DDB), a novel set of \(2d\) observables specifically designed to enable the complete characterization of any \(d\)-dimensional quantum state. These observables offer two key advantages. First, they enable direct measurement of density matrix elements without auxiliary systems, allowing any element to be extracted using only three selected observables. Second, QST for unknown rank-\(r\) density matrices—excluding only a negligible subset—can be achieved with \(O(r \log d)\) observables, significantly improving measurement efficiency. 
As for circuit implementation, each observable is iteratively generated and can be efficiently decomposed into \(O(n^4)\) elementary gates for an \(n\)-qubit system. These advances establish DDB as a practical and scalable alternative to traditional methods, offering promising opportunities to advance the efficiency and scalability of quantum system characterization.

\end{abstract}

\maketitle

\section{Introduction}

The density matrix (DM) is a fundamental representation of a quantum state \cite{fano1957descri}, essential for understanding quantum systems. Accurately determining the DM is a central challenge in quantum science. Traditionally, this challenge is addressed through quantum state tomography (QST) \cite{smithey1993measurement,paris2004quantum}, which relies on informationally complete (IC) measurements \cite{renes2004symmetric,flammia2005} and post-processing the data to estimate the quantum state. In a \(d\)-dimensional Hilbert space, a general DM contains \(d^2 - 1\) independent parameters, necessitating at least \(d^2\) projectors in an IC positive operator-valued measurement (POVM) \cite{renes2004symmetric,Caves2002}.

Projective measurements (PMs) on \(d+1\) mutually unbiased bases (MUBs)
%, which correspond to \(d(d+1)\) projectors, 
are considered optimal IC measurements \cite{Wootters1989,adamson2010}. However, the existence of such \(d+1\) MUBs for non-prime power dimensions remains an open question in quantum information theory \cite{horodecki2022five}. Moreover, in \(n\)-qubit systems (\(d = 2^n\)), estimating all \(4^n\) Pauli expectation values involves \(3^n\) unitary operations, each followed by PM on the computational basis \cite{stricker2022experimental}. This process results in \(6^n\) distinct rank-1 projectors and exponential data storage, rendering traditional QST impractical for high-dimensional systems.

Direct measurement protocols (DMPs) were developed to reduce the resources required for measurements and simplify post-processing efforts by focusing on specific DM elements rather than reconstructing the entire density matrix \cite{lundeen2011direct,bolduc2016direct,shi2015scan,pan2019direct,salvail2013full,bamber2014observing,thekkadath2016direct}. These protocols target off-diagonal DM elements, which are crucial for capturing key quantum properties such as entanglement \cite{friis2019entanglement,horodecki2009quantum} and decoherence \cite{streltsov2017colloquium,ringbauer2018certification}. DMPs often rely on weak couplings between the main system and ancillary pointers, followed by PMs. Specifically, measuring a DM element \(\rho_{jk}\) involves tailored coupling operations \(U_j\), post-selection of the state \(|k\rangle\) on the main system, and measuring different expectation values on the ancilla system \cite{gross2015novelty}. While weak measurements are minimally invasive and easy to implement, they are inherently biased and introduce unavoidable reconstruction errors, making them less precise than QST \cite{maccone2014state,gross2015novelty}. Stronger couplings have been proposed to improve accuracy \cite{vallone2016strong,calderaro2018direct,zhang2020direct}, but their practical implementation remains challenging.

The use of ancillary pointers in DMPs adds complexity to experimental systems. 
%Achieving precise coupling between the pointer and the main system becomes increasingly challenging. 
%, particularly as the system size grows. 
Initially, a single pointer sufficed for pure states (rank-1 DMs) \cite{lundeen2011direct}, but general DMs require another pointer \cite{lundeen2012procedure,thekkadath2016direct}. 
Efforts to eliminate pointers have led to various approaches. For example, \(\delta\)-quench measurements achieve pointer-free operation for pure states in specific systems \cite{zhang2019delta}, while phase-shifting techniques enable pointer-free measurements but require \(O(d^2)\) unitary operations \cite{feng2021direct}. For multi-qudit DMs, it has been shown that a single pointer can suffice \cite{xu2024resource}. Strong measurement-based DMP \cite{calderaro2018direct} discussed the direct reconstruction of all DM elements using \(d^2\) projectors in QST without pointers, but it could introduce greater experimental complexity compared to \(O(d)\) strong coupling operations and PMs needed in DMPs.

% These raise a question: can QST of any finite dimension be achieved without ancillary pointers, using only \(O(d)\) unitary operations and PM on the computational basis? Moreover, is it possible for a constant number of these operations to directly measure each DM element?

These raise a question: Can DMP benefit the performance of QST? Moreover, is there a protocol that advances both DMP and QST? 

%In this work, we achieve this by designing a set of \(2d\) eigenbases for QST, corresponding to \(2d\) unitary operations followed by PM on the computational basis. Unlike the minimal \(d+1\) MUBs, the existence of these observables for any dimension \(d\) is ensured through a deterministic construction algorithm. This set of observables also functions as a DMP, enabling each DM element to be measured with just three observables, thereby ensuring accuracy without the need for ancillary systems. Furthermore, these observables are applied to perform QST on quantum states with prior knowledge of rank-\(r\), connecting through DMPs via matrix completion techniques. As a result, we demonstrate that for a rank-\(r\) density matrix, only \(O(r \log d)\) of the \(2d\) observables are required for full-state characterization. This approach dramatically reduces the required unitary operations, showing an exponential decrease compared to the \(O(rd \log^2 d)\) operations needed when using random Pauli observables from the \(d^2\) set via compressed sensing \cite{Gross2010}. Finally, we present a unified formula that expresses all these unitary operations, where each can be decomposed into a permutation gate followed by Pauli measurements. The permutation gate itself can be efficiently decomposed into \(O(n^4)\) gates on an \(n\)-qubit system. This not only enhances the efficiency of quantum state learning but also opens new pathways for practical implementations in high-dimensional quantum systems. 

In this work, we design a set of \(2d\) eigenbases for QST, corresponding to \(2d\) unitary operations followed by PM on the computational basis. Unlike the minimal and optimal \(d+1\) MUBs, these eigenbases are guaranteed to exist for any dimension \(d\) through a deterministic construction algorithm. Serving as a DMP, each DM element can be directly measured using three eigenbases from this set, eliminating the need for ancillary systems and ensuring accuracy through strong measurements. 
Building on this foundation, we apply these observables to perform QST on quantum states with prior knowledge of rank-\(r\), utilizing matrix completion techniques. For rank-\(r\) density matrices, only \(O(r \log d)\) out of the \(2d\) observables are required for full characterization, except for a set of zero-measure. This approach significantly reduces the operational cost compared to the \(O(rd \log^2 d)\) operations needed with random Pauli observables \cite{Gross2010}. For \(n\)-qubit systems, each unitary operation is represented by a permutation gate followed by one of \(2n\) special Pauli measurements, with the permutation gate efficiently implemented in \(O(n^4)\) gates. 
This method not only improves the efficiency of quantum state learning but also offers a scalable and practical framework for high-dimensional quantum systems, paving the way for broader applications in quantum information science.

\section{Preliminaries and Dense Dual Bases}

\textit{Observables and Projective Measurements.---}
When measuring a DM \(\rho\) with an observable \(O = \sum_{k=1}^d \lambda_k |O_k\rangle \langle O_k|\), the Born rule states that the measurement outcome \(\lambda_k\) occurs with probability \(\mathrm{tr}(\rho |O_k\rangle \langle O_k|)\). Each observable corresponds to one PM onto its eigenbasis \(\{|O_k\rangle\}_{k=1}^d\), or equivalently, to one unitary operation \(U^\dagger\) followed by PM on the computational basis \(\{|k\rangle\}_{k=1}^d\), where \(U = \sum_{k=1}^d |O_k\rangle \langle k|\). This equivalence ensures \(\mathrm{tr}(U^\dag \rho U |k\rangle \langle k|) = \mathrm{tr}(\rho |O_k\rangle \langle O_k|)\).

For \(n\)-qubit systems, the \(4^n\) Pauli observables, commonly used for QST, correspond to \(4^n\) unitary operations:
\begin{equation}
\rho = \frac{1}{2^n} \sum_{i_1,\cdots,i_n=0}^3 \text{tr}(\rho \sigma_{i_1} \otimes \cdots \otimes \sigma_{i_n}) \sigma_{i_1} \otimes \cdots \otimes \sigma_{i_n}.
\end{equation}
Experimentally, \(3^n\) distinct unitary operations (excluding \(I\) for each qubit) and PM on the computational basis suffice to extract the \(4^n\) expectation values \cite{Nielsen2002}. However, this process involves \(6^n\) distinct projectors, leading to significant redundancy compared to the \(4^n\) actually required. For arbitrary dimension \(d\), the Pauli observables generalize to \(d^2\) operators such as Gell-Mann or Heisenberg-Weyl matrices. 
%These highlight the need for efficient measurement strategies with $O(d)$ unitary operations and PM onto the computational basis in high-dimensional systems.

A promising alternative is PMs onto \(d+1\) MUBs, recognized as the minimal and optimal strategy for QST. 
Two orthonormal bases $\{|a_j\rangle\}_{j=1}^d$ and $\{|b_k\rangle\}_{k=1}^d$ are termed as mutually unbiased if $ \lvert \langle a_j \lvert b_k \rangle \rvert^2 = 1/d$ for all $j,k$. 
MUBs are widely used in quantum information applications \cite{maassen1988generalized,cerf2002security,ballester2007entropic,giovannini2013characterization,tavakoli2021mutually} and can be efficiently implemented in \(n\)-qubit systems with \(2^n+1\) circuits and PM on the computational basis \cite{seyfarth2011construction,seyfarth2015practical,yu2023effi}. 
However, for arbitrary dimension \(d\), the existence of \(d+1\) MUBs is only guaranteed for prime power dimensions. In non-prime power cases (\(d=6\), for example), strong numerical evidence suggests that only three MUBs exist \cite{brierley2009constructing,goyeneche2013mutually}. Furthermore, while MUBs are highly efficient for QST, they are not designed to directly measure individual DM elements. This limitation, combined with the uncertainty surrounding their existence for general \(d\), underscores the need for alternative strategies.

 For any dimension \(d\), Caves, Fuchs, and Schack considered using \(d^2\) rank-1 projections to directly determine all DM elements \cite{Caves2002}. The projected states are defined as:
\begin{equation}
	\mathcal{A}_d=\{|l\rangle,|\phi_{jk}^{+}\rangle,|\psi_{jk}^{+}\rangle: 0\le j<k\le d-1;~ l\in [d]\}.
	\label{dd elements}
\end{equation}
where \(|\phi_{jk}^{\pm}\rangle \doteq (|j\rangle \pm |k\rangle)/\sqrt{2}\) and \(|\psi_{jk}^{\pm}\rangle \doteq (|j\rangle \pm i|k\rangle)/\sqrt{2}\). Using at most four projectors, each DM element \(\rho_{ij}\) can be reconstructed:
\begin{eqnarray}\label{element_extract}
&&\rho_{ll}=\mbox{tr}(\rho|l\rangle\langle l|),\nonumber\\
&&\rho_{jk}=\mbox{tr}(\rho(|\phi_{jk}^{+}\rangle\langle\phi_{jk}^{+}|-i|\psi_{jk}^{+}\rangle\langle\psi_{jk}^{+}|))-\frac{1-i}{2}(\rho_{kk}+\rho_{jj}). \nonumber\\
\end{eqnarray} 
While this approach ensures informational completeness and direct reconstruction of DM elements, implementing \(O(d^2)\) projectors could be more challenging than $O(d)$ coupling operations and PMs in DMP  
\cite{calderaro2018direct}, especially for large $d$.

\section{Results and Analysis}
\emph{DDBs}---
We begin with an efficient algorithm to construct at most $2d$ eigenbases to cover the states in Eq. (\ref{dd elements}) for arbitrary dimension \(d\). 
These are referred to as dense dual bases (DDBs).
\begin{theorem}
For an arbitrary dimension \(d\), all elements of a density matrix can be directly measured using \(2d-1\) DDBs for even \(d\) and \(2d\) DDBs for odd \(d\).
\label{result1}
\end{theorem}

\textbf{Analysis.} Each designed DDB must satisfy two key constraints: \textit{orthogonality} (elements within a basis must be orthogonal) and \textit{completeness} (each basis must contain \(d\) elements). To achieve completeness, we extend the set \(\mathcal{A}_d\) in Eq.~(\ref{dd elements}) by adding additional elements \(\{|\phi_{jk}^{-}\rangle\}\) and \(\{|\psi_{jk}^{-}\rangle\}\), resulting in \(2d^2-d\) elements.

For \(d=2\), these six elements  
\begin{equation}
    \{|0\rangle, |1\rangle, |\phi_{01}^{+}\rangle, |\phi_{01}^{-}\rangle, |\psi_{01}^{+}\rangle, |\psi_{01}^{-}\rangle\}
\end{equation}
are the eigenstates of the Pauli observables \(Z, X, Y\). The eigenbases \(\{|\phi_{01}^{\pm}\rangle\}\) and \(\{|\psi_{01}^{\pm}\rangle\}\) correspond to \(X\) and \(Y\), respectively. For general \(d\), we arrange \( |\phi_{jk}^{\pm}\rangle\) and \( |\psi_{jk}^{\pm}\rangle\) into separate eigenbases, ensuring orthogonality by preventing overlapping pairs \(\{|\phi_{j_1k_1}^{\pm}\rangle, |\phi_{j_2k_2}^{\pm}\rangle\}\) or \(\{|\psi_{j_1k_1}^{\pm}\rangle, |\psi_{j_2k_2}^{\pm}\rangle\}\) with shared indices \(j_1 = j_2\) or \(k_1 = k_2\).

The problem reduces to a combinatorial optimization: construct all pairs \(\{(j, k) : 0 \leq j < k \leq d-1\}\) using minimal bands, where each band contains \(d\) numbers $\{0,1\cdots,d-1\}$. For even \(d\), each band forms \(  d/2  \) pairs; for odd \(d\), each band forms \((d-1)/2\) pairs and includes a single remaining element.

For even \(d\), at least \(C_d^2/(d/2) = d-1\) bands are required, corresponding to \(2(d-1)\) eigenbases. Adding the computational basis \(\mathcal{B}_0 = \{|0\rangle, \cdots, |d-1\rangle\}\), the total number of DDBs is \(2d-1\). For odd \(d\), \(d\) bands are needed, and the single remaining elements cover the computational basis, resulting in \(2d\) DDBs.

General $n$-qubit case. 
For \(d = 2^n\), we illustrate the minimal $2^n-1$ partitions of bands with $n$ iterations in Fig. (\ref{p+a}). 
\begin{figure}[!htb]
  \begin{center}
    \includegraphics[width=0.45\textwidth]{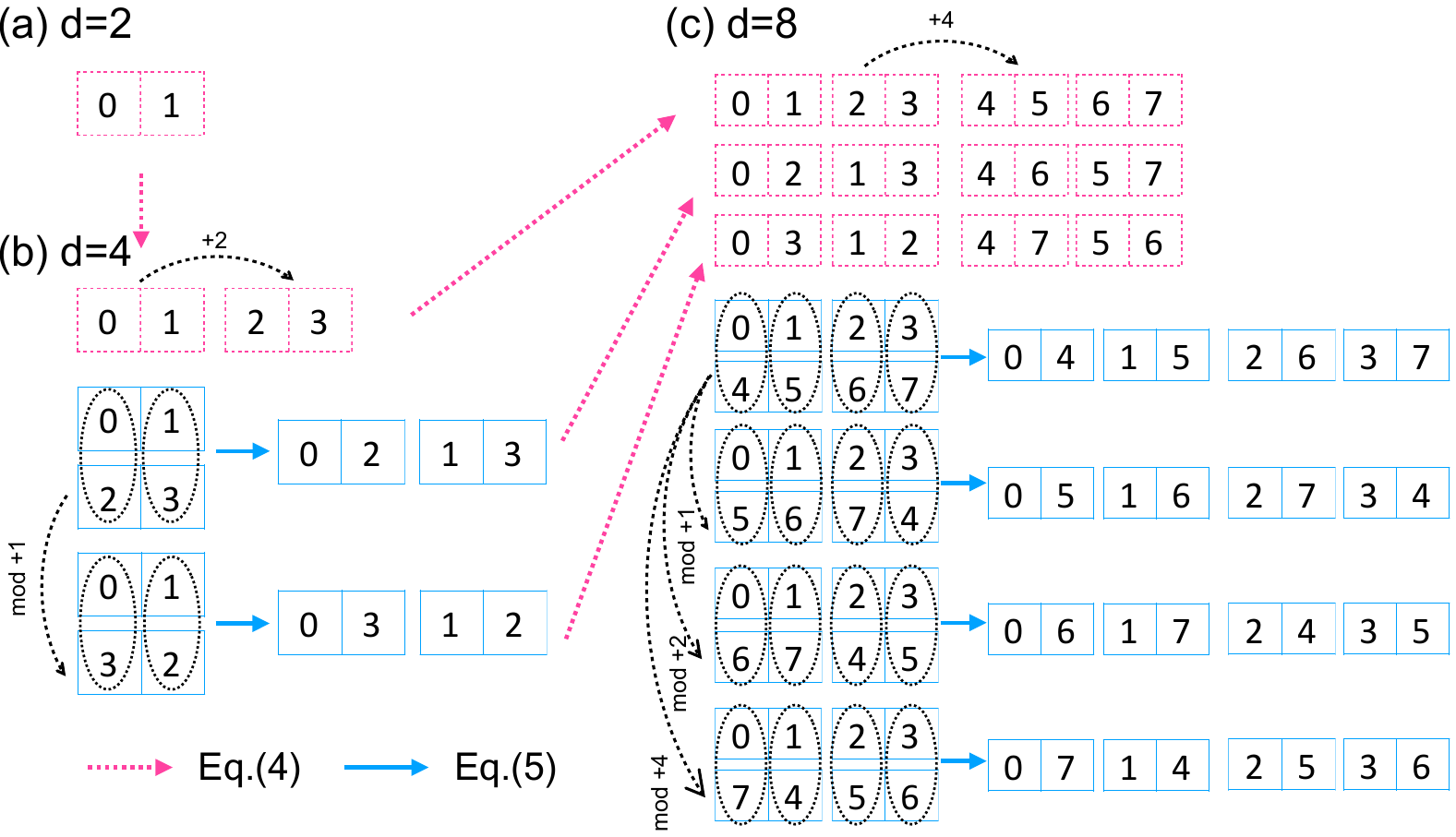}
    \caption{Optimal \(d-1\) cuttings for (a) \(d=2\), (b) \(d=4\), and (c) \(d=8\). The \(\bmod + x\) in the figure is shorthand for the calculation in Eq. (\ref{eq2}).}
 \label{p+a}
    \end{center}
\end{figure} 
For \(d=4\), the seven DDBs, corresponding to three partitions, are:
\begin{align*}
  \mathcal{B}_0^4 &= \{|0\rangle,|1\rangle,|2\rangle,|3\rangle\}, \\
  \mathcal{B}_1^4 &= \{ |\phi_{01}^{\pm}\rangle,|\phi_{23}^{\pm}\rangle\},  ~~
  \mathcal{C}_1^4 = \{ |\psi_{01}^{\pm}\rangle,|\psi_{23}^{\pm}\rangle\}, \\
  \mathcal{B}_2^4 &= \{ |\phi_{02}^{\pm}\rangle,|\phi_{13}^{\pm}\rangle\},  ~~
  \mathcal{C}_2^4 = \{ |\psi_{02}^{\pm}\rangle,|\psi_{13}^{\pm}\rangle\}, \\
  \mathcal{B}_3^4 &= \{ |\phi_{03}^{\pm}\rangle,|\phi_{12}^{\pm}\rangle\},  ~~
  \mathcal{C}_3^4 = \{ |\psi_{03}^{\pm}\rangle,|\psi_{12}^{\pm}\rangle\}.
\end{align*}

The partitions are constructed iteratively:
\begin{itemize}
    \item Base case (\(n=1\)): 
    \(
    \mathbb{T}^2 = \{(0,1)\}.
    \)
    \item  Recursive step (\(n > 1\)):
    \begin{itemize}
        \item  Merging partitions. Merge partitions from the \((n-1)\)-qubit case:
        \begin{equation}\label{eq1}
            T^{2^n}_{m_1} = T^{2^{n-1}}_{m_1} \cup (T^{2^{n-1}}_{m_1} + 2^{n-1}), \quad 1 \leq m_1 \leq 2^{n-1}-1.
        \end{equation}
        \item Crossed partitions. Define the left \(2^{n-1}\) partitions directly:
        \begin{equation}\label{eq2}
            T^{2^n}_{m_2} = \{(k, 2^{n-1} + (k + m_2) \bmod 2^{n-1}) : k \in [2^{n-1}]\},
        \end{equation}
         where $2^{n-1} \leq m_2 \leq 2^n - 1$. 
    \end{itemize}
\end{itemize}

For general even and odd $d$, an iterative algorithm guarantees complete coverage of all pairs \(\{(j, k) : 0 \leq j < k \leq d-1\}\) with minimal cost in terms of bands. Details are in Appendix A. \qed

Without ancillas, the minimal eigenbases for QST are \(d+1\) MUBs (if they exist), requiring \(d(d+1)\) projectors \cite{Wootters1989}. In contrast, DDBs are applicable for all \(d\) and produce at most \(2d^2\) projectors. Earlier DMPs \cite{lundeen2012procedure,calderaro2018direct} required four or three Pauli observables on two-pointers, leading to total \(16d^2\) or \(8d^2\) projectors, respectively. In comparison, \(2d\) (or \(2d-1\)) DDBs minimize the number of projectors needed to directly measure all DM elements. Notably, three DDBs suffice to determine each DM element, capturing both \(d\) diagonal elements and \(d/2\) off-diagonal elements.

The proposed construction algorithm uses \(\log d\) iterations to generate minimal partitions, offering scalability for large dimensions. While alternative combinatorial optimization methods, such as brute-force search, exist, the proposed approach significantly improves efficiency.

For \(d=6\) (qubit-qutrit systems), where only three MUBs exist and a fourth has not been found \cite{brierley2009constructing,raynal2011mutually}, incomplete MUBs are insufficient for QST. By contrast, 11 DDBs reconstruct all \(36 \times 36\) DM elements, as verified through numerical simulations in Appendix B.

\emph{Applications}---An important approach to reducing QST measurement resources is leveraging prior knowledge. This includes matrix product states \cite{Cramer2010,Lanyon2017}, permutation-invariant states \cite{Toth2010}, and low-rank states \cite{Gross2010,Oren2017}. For instance, compressed sensing shows that randomly selecting \(O(rd \log^2 d)\) Pauli observables suffices to recover a rank-\(r\) DM with high probability \cite{Gross2010}.

When \(r=1\), rank-1 DMs correspond to pure states. PMs onto \(3d-2\) states can uniquely determine all pure states, except for a measure-zero set \cite{flammia2005, Wang2018}. These states can be \(\{|l\rangle, |\phi_{jk}^{+}\rangle, |\psi_{jk}^{+}\rangle : 0 \le j < k \le d-1, |j-k| \le 1; l \in [d]\}\). The five eigenbases for pure state tomography constructed by Goyeneche et al. \cite{Goyeneche2015} form a minimal cover of these \(3d-2\) states and are proven rank-1 strictly complete \cite{Baldwin2016}. Rank-\(r\) strictly complete measurements uniquely determine rank-\(r\) states with high probability. No other physical states share the same measurement probability distributions, except for a dense set of rank-\(r\) states on a measure-zero set.

Using DDB measurements, three DDBs suffice to determine any specific DM element, while all \(2d\) DDBs are necessary to reconstruct the entire set of DM elements. With prior knowledge, matrix completion techniques \cite{jain2013low,koltchinskii2015optimal,acharya2016statistically,Baldwin2016} enable efficient reconstruction of the entire DM from a subset of key DM elements. This naturally leads to the question of how partial DDBs can be employed to reconstruct rank-\(r\) DMs.

\begin{theorem}
To uniquely determine a rank-\(r\) DM of dimension \(d\), \(O(r \log \frac{d}{r})\) DDBs suffice, except on a measure-zero set. Here, the rank \(r\) is assumed to be significantly smaller than the dimension \(d\) (\(r \ll d\)).
\label{result2}
\end{theorem}

\textbf{Analysis.}  
Reconstructing a rank-\(r\) DM can reduce to measuring specific elements selected based on the following label set: 
\begin{equation}
    C = \{(j,k) : |j-k| \le r\}.
\end{equation}
Fig.~\ref{reconstruction} illustrates the elements involved for \(r=2\). This set includes all elements along the main diagonal and the adjacent diagonals up to the \(r\)-th order, capturing the critical information for reconstruction.
\begin{figure}[!htb]
	\begin{center}
	\includegraphics[width=0.28\textwidth]{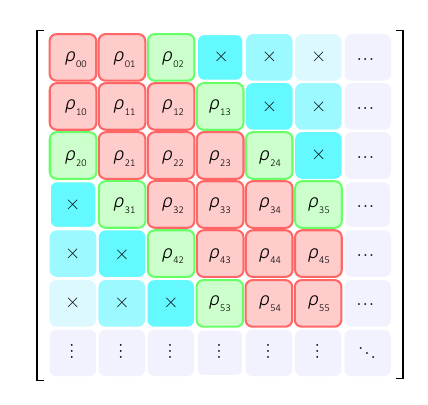}
		\caption{Diagonal DM elements to be measured. For \(r=2\), elements in the red region (\(|j-k| \le r-1\)) form the principal submatrix \(A_k\), while the light green region represents elements with \(|j-k| = r\).}
		\label{reconstruction}
	\end{center}
\end{figure}

The complete DM can be recovered using convex optimization, formulated as:
\begin{equation}
\begin{aligned}
    \hat{X} = &\arg \min_{X} \left\| \mathrm{tr}(X |k\rangle\langle j|) - \rho_{jk} \right\| \quad \text{for } (j,k) \in C, \\
    &\text{such that } X \succeq 0, \mathrm{tr}(X) = 1.
\end{aligned}
\label{eq:optimization_problem}
\end{equation}
Previous research has proved that a POVM capable of determining the elements in \(C\) is classified as rank-\(r\) strictly complete \cite{Baldwin2016}. These measurements are highly advantageous as they ensure efficient recovery through convex optimization and are robust against noise and state preparation errors.

Thus we should pick out DDBs containing the following states:
\begin{equation}\label{eq:partial_elements}
    \mathcal{B} = \{|l\rangle, |\phi_{jk}^{+}\rangle, |\psi_{jk}^{+}\rangle : 0 \le j < k \le d-1, |j-k| \le r; l \in [d]\}.
\end{equation}
Equivalently, the task reduces to finding the partitions that contain pairs \((j,k)\) where \(|j-k| \le r\) and $0\le j<k \le d-1$. The PMs onto these DDBs can determine the target DM elements.

For \(d = 2^n\) or general \(d\), the required elements can be covered by \(O(r \log \frac{d}{r})\) DDB partitions. As detailed in Appendix C, this iterative approach ensures scalability and efficiency, making it suitable for high-dimensional systems.  \qed

An alternative approach to cover the states in Eq.~(\ref{eq:partial_elements}) employs \(4r+1\) eigenbases, specifically applicable for dimensions \(d = 2^n\) \cite{Baldwin2016}. While Result~\ref{result2} extends to arbitrary \(d\), it requires a larger number of eigenbases, reflecting a trade-off between generality and simplicity.

Compressed sensing, based on randomly sampled Pauli observables, estimates unknown states with high probability using \(O(rd \log^2 d)\) expectation values \cite{Gross2010}. Although effective, this method requires \(O(rd \log^2 d)\) separable unitary operations, each followed by PM on the computational basis. By contrast, DDBs achieve an exponential reduction in the types of unitary operations required, needing only \(O(r \log d)\) entangled unitary operations. The frequent changes of unitary operation in the measurement setup could introduce more noise and unrelated errors, potentially reducing the overall accuracy. However, DDB operations necessitate recording all \(d\) measurement outcomes for each computational basis measurement. When applying Eq.~(\ref{eq:optimization_problem}) with DDBs to reconstruct rank-\(r\) DMs, the required post-processing data volume scales as \(O(rd)\), slightly reduced compared to \(O(rd\log^2 d)\) in compressed sensing, but still dependent on \(d\). 

DDB measurements are \(d\)-outcome measurements, whereas \(n\)-qubit Pauli measurements, which record only expectation values, can be treated as 2-outcome measurements. While the sampling complexity of compressed sensing with Pauli measurements has been rigorously analyzed \cite{flammia2012quantum}, deriving the specific sampling complexity for \(d\)-outcome DDB measurements under a given estimation fidelity remains an interesting direction for future work.

The failed set of determinations in matrix completion can be characterized by \(r \times r\) principal submatrices:
\begin{equation}        
	A_k=\left(                  
	\begin{array}{ccc}   
		\rho_{k,k} & \cdots & \rho_{k,k+r-1}\\
		\vdots & \ddots & \vdots \\
		\rho_{k+r-1,k} &  \cdots  & \rho_{k+r-1,k+r-1} 
	\end{array}
	\right),                
\end{equation}
where \(k=0, \dots, d-r\). Failures occur if \(A_i\) is singular for \(i = 0, \ldots, d - r - 1\) and simultaneously \(A_j\) is singular for \(j = i \pm 1\) \cite{Baldwin2016}.

An adaptive strategy can handle certain failure cases, such as when some diagonal elements are zero. For example, if \(\rho_{00} = 0\), all \(|\psi_k\rangle\) must have zero components in the first basis state, leading to \(\rho_{0l} = \rho_{l0} = 0\) for all \(l = 0, \dots, d-1\). In such cases, the submatrix can be reconstructed by erasing the first row and column of \(\rho\). The DDBs can then be redesigned for the reduced subspace spanned by \(\{|1\rangle, \dots, |d-1\rangle\}\), ensuring robust reconstruction in the presence of such failures.

To test the results, we process a numerical simulation, which explores the fidelity of reconstructed quantum density matrices under varying numbers of measurements and ranks. 
The simulation begins with random quantum density matrices. For each matrix, multiple reconstructions are performed using different numbers of measurements via DDBs. The fidelity of the reconstructed matrices is then calculated, with the process repeated 20 times to obtain average fidelity values.
Figure \ref{reconstruction2}(a) shows how the fidelity changes with the number of measurements for a fixed dimension ($d=16$) while varying the rank of the density matrix ($r=2,4,8,16$), demonstrating that increasing the number of measurements generally improves the fidelity of the reconstruction. In addition, we conducted a test on reconstructing density matrices using either DDBs or a Pauli-based compressed sensing method. The compressed sensing involved an initial measurement with varying numbers of random Pauli bases (CSPs), specifically $rn$, $10rn$, $drn^2$, and $10drn^2$, followed by convex optimization to minimize the nuclear norm of the density matrix. In contrast, Result \ref{result2} employed $\mathcal{O}(rn)$ DDBs with semidefinite optimization to minimize the Frobenius distance. As shown in Figure \ref{reconstruction2}(b)-(d), corresponding to $r=2$ and $d=4, 8, 16$, our methods demonstrate clear advantages in requiring fewer measurements.

 \begin{figure}[!htb]
	\begin{center}
	\includegraphics[width=0.46\textwidth]{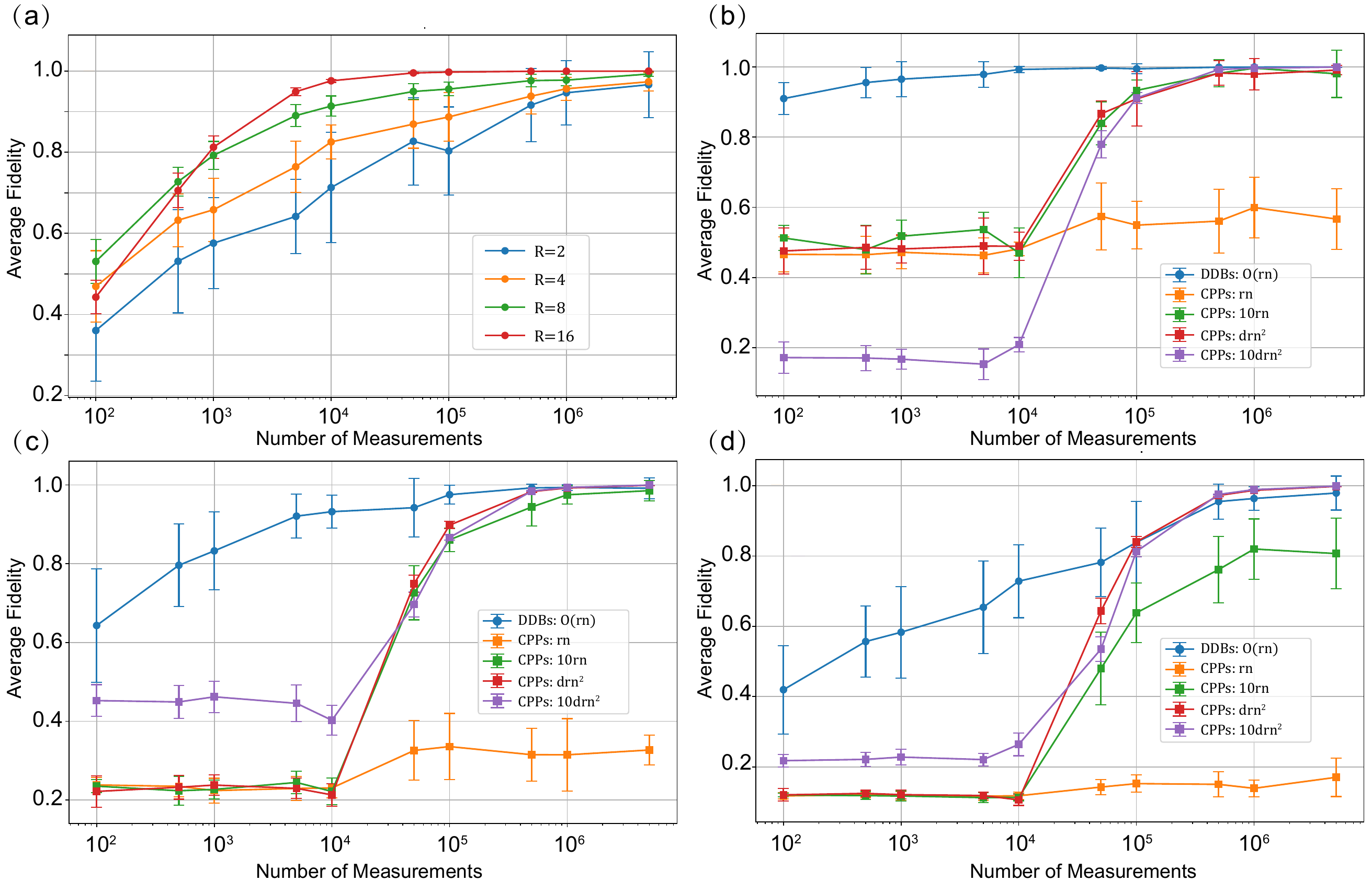}
		\caption{Simulation results of Result \ref{result2}. (a) illustrates the fidelity's dependence on the number of measurements for $d=16$ while varying the rank of the density matrix. (b)-(d) compare the fidelity variations across different numbers of measurements between the methods using DDBs(Result \ref{result2}) and Pauli-based compressed sensing method, with $d=4, 8, 16$.} 
		\label{reconstruction2}
	\end{center}
\end{figure}

\emph{Implementation}---The circuit implementation of DDBs in $n$-qubit systems is detailed below, highlighting the iterative construction approach and its practical scalability.

\begin{theorem}
In \(n\)-qubit systems, the projective measurements onto the \(2 \times 2^n - 1\) DDBs can be implemented using permutational operations followed by Pauli measurements. Each permutation can be decomposed into \(O(n^4)\) elementary gates.

The PM onto computational basis \(\{|0\rangle, \cdots, |2^n-1\rangle\}\) corresponds to the Pauli observable \(Z_1 Z_2 \cdots Z_n\) (with \(\otimes\) omitted for convenience). 

The PMs onto the  remaining \(2 \times 2^n - 2\) DDBs are performed  using a permutational operation followed by one of \(2n\) Pauli observables:
\begin{equation}\label{equ:2npauli}
    Z_1 \cdots Z_{j-1} X_j Z_{j+1} \cdots Z_n ~ \text{or} ~ Z_1 \cdots Z_{j-1} Y_j Z_{j+1} \cdots Z_n, 
\end{equation}
where \(j = 1, \cdots, n\).

The corresponding permutation operation is:
\begin{equation}\label{equ:permutation}
    P_{j,k} = I^{\otimes (j-1)} \otimes \left( |0\rangle\langle 0| \otimes I^{\otimes (n-j)} + |1\rangle\langle 1| \otimes (\mathcal{U}_{n-j})^k \right),
\end{equation}
with \(k = 0, \cdots, 2^{n-j} - 1\), and 
\begin{equation}
    \mathcal{U}_{n-j} = \sum_{m=0}^{2^{n-j}-1} |m - 1 \bmod 2^{n-j}\rangle \langle m|.
\end{equation}

\label{result3}
\end{theorem}

Each of the \(2n\) Pauli observables in Eq. (\ref{equ:2npauli}) can be implemented by applying a Hadamard gate \(H\) (or \(\tilde{H}^{\dag} = S^{\dag}H\)) to qubit \(j\), followed by a computational basis measurement \(Z_1 Z_2 \cdots Z_n\). Thus, the PM onto each DDB is effectively transformed into a unitary operation followed by a PM on the computational basis.

For \(j = 1, \cdots, n\), there are \(2^{j-1}\) types of permutational operations \(P_{j,k}\) as defined in Eq. (\ref{equ:permutation}). Therefore, the total number of distinct permutational operations is \(2^n - 1\). Including their dual forms and the computational basis, these projected bases correspond one-to-one to the \(2 \times 2^n - 1\) DDBs.

Each \(P_{j,k}\) can be regarded as an \((n-j+1)\)-qubit operation with one control qubit. The operation \(\mathcal{U}^k_{n-j}\) is defined as applying \(\mathcal{U}_{n-j}\) over \(k\) times, where \(k = 1, \cdots, 2^{n-j} - 1\). Even when \(k\) is exponentially large, \(\mathcal{U}^k_{n-j}\) can be efficiently implemented as some linear combinations of \(\mathcal{U}^{2^0}_{n-j}, \mathcal{U}^{2^1}_{n-j}, \cdots, \mathcal{U}^{2^{n-j-1}}_{n-j}\). 
Interestingly, the circuit decomposition of \(\mathcal{U}^{2^l}_{n-j}\) becomes more efficient than \(\mathcal{U}^{2^m}_{n-j}\) for \(l \geq m\). The classical counterpart of \(\mathcal{U}_{n-j}\) is a basic shift operation. Its quantum version, \(\mathcal{U}_{n-j}\), corresponds to the shift operator in the Weyl-Heisenberg group.

By leveraging Corollary 7.6 in \cite{barenco1995elementary}, each generalized \(n\)-qubit Toffoli gate can be decomposed into \(O(n^2)\) elementary quantum gates. The operation \(\mathcal{U}_{n-j}\) is constructed as a sequence of controlled operations, specifically generalized Toffoli gates with varying control sizes. 
Consequently, \(\mathcal{U}_{n-j}\) can be realized using \(O((n-j)^3)\) elementary gates. More generally, any permutation operation \(P_{j,k}\) can be decomposed into at most \(O(n^4)\) elementary gates. Thus, the PM onto each DDB can be implemented by applying a permutational operation decomposed into \(O(n^4)\) elementary gates, followed by a computational basis measurement.

The detailed proof is provided in Appendix D \cite{supp}.

\section{Conclusion and discussion} 
 
\begin{table*}[t]
\centering
 \begin{tabular}{c c c c c c} 
 \hline
 \hline
$n$-qubit IC measurements & Number & Directness  & Generalized dimensions $d$ & Rank-$r$ QST & Circuits \\
\hline
Pauli observables&  $4^n$& No & Yes & $O(r n^2\cdot 2^n)$ \cite{Gross2010}& Local  \\
MUBs&  $2^n+1$ & No & Open question \cite{horodecki2022five} &  --& \cite{seyfarth2011construction,seyfarth2015practical,yu2023effi} \\
DDBs & $2^{n+1}-1$ & Yes &Result~\ref{result1}, yes  & Result~\ref{result2}, $O(rn)$  & Result~\ref{result3}  \\ 
\hline
\hline
\end{tabular}
\caption{Comparison of \(n\)-qubit IC measurements including Pauli observables, MUBs, and DDBs for QST. When \(n=1\), all measurement strategies yield identical results. DDBs stand out by enabling the direct reconstruction of DM elements using a constant number of eigenbases (up to three). In contrast, MUBs, while theoretically minimal and optimal, face the unresolved question of their existence in all dimensions \(d\), rendering them insufficient for complete QST. Pauli observables, while simple in terms of circuit implementation, require the largest number of operations.
}
\label{table1}
\end{table*}
 
Efficient characterization of quantum states remains a central challenge in quantum science. In this work, we present a method for constructing and decomposing minimal DDBs, offering a scalable solution for QST across arbitrary dimensions \(d\). Compared to traditional approaches such as Pauli observables and the theoretically optimal \(d+1\) MUBs (Table \ref{table1}), our method employs direct measurement protocols capable of reconstructing each DM element using a constant number of unitary operations. By leveraging strong measurements, our scheme ensures measurement accuracy, eliminates the need for ancillary pointers, and reduces the number of required projectors to the minimal \(2d^2\). 
Notably, we demonstrate that \(O(d)\) unitary operations and computational basis measurements suffice for reconstructing all DM elements, achieving an exponential reduction in unitary operation requirements for low-rank DMs while also slightly decreasing the data required for post-processing than random Pauli measurements. The practical feasibility of our approach is validated through extensive numerical simulations and cloud-based quantum experiments, highlighting its potential for real-world applications.

Building on this work, several promising research directions emerge. One exciting avenue is the application of random DDBs for classical shadow tomography, which could enable constant-time post-processing \cite{wang2024quantum} for any observable in a single experiment—dramatically improving upon the worst exponential complexity of random Clifford measurements \cite{huang2020predicting}. This advancement holds promise for more fidelity estimation and entanglement detection tasks. 
Another important challenge is optimizing the \(O(n^4)\) complexity of permutation operations in DDBs, which could expand their utility in high-dimensional quantum systems. Extending these methods to \(d\)-level systems further raises opportunities to identify efficient physical implementations. Additionally, integrating DDBs with matrix recovery techniques in QST presents exciting possibilities: with prior rank-$r$ information, our method achieves exponential reductions in unitary operations compared to the full dimension $d$.  It remains an interesting question whether other forms of prior knowledge could yield similar gains. 
Finally, comparative studies of Pauli observables, MUBs, and DDBs under experimental noise and across different platforms could provide crucial insights for their practical adoption. These directions underscore the theoretical and practical potential of DDBs in advancing quantum information science.

\textbf{Acknowledgments---}
We thank Cheng Qian for the brute-force searches that helped clarify the structure of the DDBs construction. This work was supported by the National Natural Science Foundation of China under Grants No. 62001260, 42330707 (Y.W.), 11701536 (Y.L.), and 11905111 (K.L.); the Beijing Natural Science Foundation under Grant No. Z220002 (Y.W.); and the Major Key Project of PCL (Y.L.).

\textbf{Author Contributions---}  
Y. W. conceived the idea for this paper, applied the rank-\( r \) QST, decomposed the circuits, and drafted the manuscript. H. J. constructed partitions for \( d = 2^n \), while Y. L. handled cases where \( d/2 \) is odd. K. L. conducted the error analysis, numerical simulations, and cloud experiments, and provided revisions to the manuscript. All authors reviewed and approved the final version of the manuscript for submission.

\bibliographystyle{apsrev4-1}
\bibliography{sample.bib}

\clearpage
\onecolumngrid
\appendix 
\section{Detailed proofs in Result 1}

In the main text, we demonstrated the construction of \(d-1\) partitions for \(d=2^n\), ensuring that each pair \((j, k)\) with \(0 \leq j < k \leq 2^n - 1\) is included in one partition. Here, we extend this construction to arbitrary even dimensions \(d\).

\begin{lemma}\label{lemma1}
For even \(d\), we can construct \(d-1\) partitions \(\{T^d_1, \cdots, T^d_{d-1}\}\), ensuring that every pair \((j, k)\) with \(0 \leq j < k \leq d - 1\) is is in one of the partitions.
\end{lemma}

\begin{proof}
We use mathematical induction on \(d\), where \(d\) is even.

\textbf{Base Case (\(d=2\)):}  
The single pair \((0, 1)\) corresponds to one partition:  
\[
\mathbb{T}^2 = \{T_1^2\} = \{(0, 1)\}.
\]  
Thus, the lemma holds for \(d=2\).

\textbf{Inductive Step:}  Since \(d\) is even, one of the numbers \(d/2\) or \(d/2+1\) is also even. 
Assume the lemma holds for \(d/2\) or \(d/2 + 1\). We construct \(d-1\) partitions for even \(d\).

\textbf{Case 1: \(d/2\) is even.}  
By the inductive hypothesis, there exist \(d/2 - 1\) partitions \(\{T_t^{d/2} : t = 1, 2, \cdots, d/2 - 1\}\) that cover all pairs in \(0 \leq j < k \leq d/2 - 1\). Using these, we construct \(d-1\) partitions for \(d\) as follows:

1. \textit{Merged Partitions:}
\begin{equation}\label{merge1}
T_t^d = T_t^{d/2} \cup (T_t^{d/2} + d/2), \quad t = 1, 2, \cdots, d/2 - 1.
\end{equation}

2. \textit{Crossed Partitions:}
\begin{equation}\label{insect1}
T_t^d = \{(j, d/2 + [(t + j) \bmod d/2]) : 0 \leq j \leq d/2 - 1\}, \quad t = d/2, \cdots, d-1.
\end{equation}

\textbf{Case 2: \(d/2\) is odd.}  
By the inductive hypothesis, there exist \(d/2\) partitions \(\{T_t^{d/2+1} : t = 1, 2, \cdots, d/2\}\) for even dimension $d/2+1$. Using these, we construct \(d-1\) partitions for \(d\) as follows:

1. \textit{Modified Partitions:}
From these $d/2$ partitions, we select neighbors of \(d/2\) as \(\{c_t : t = 1, \cdots, d/2\}\) and define:
\begin{equation}\label{merge2}
T_t^d = T_t^{d/2+1} \cup (T_t^{d/2+1} + d/2) - \{(c_t, d/2), (d/2 + c_t, d)\} \cup \{(c_t, d/2 + c_t)\}.
\end{equation}

2. \textit{Crossed Partitions:}
\begin{equation}\label{insect2}
T_t^d = \{(j, d/2 + [(t + j) \bmod d/2]) : 0 \leq j \leq d/2 - 1\}, \quad t = d/2 + 1, \cdots, d-1.
\end{equation}

\textbf{Example for \(d=6\):}  
Fig.~\ref{partitions6} illustrates the construction for \(d=6\), where neighbors of \(3 = d/2\) are \(2, 1, 0\). For instance:
\[
T_1^6 = \{(0, 1), (2, 5), (3, 4)\}.
\]  
This is constructed as:
\[
T_1^6 = (0, 1) \cup (2, 3) \cup (3, 4) \cup (5, 6) - \{(2, 3), (5, 6)\} \cup \{(2, 5)\}.
\]

\begin{figure}[htb]
\centering
\includegraphics[width=0.6\textwidth]{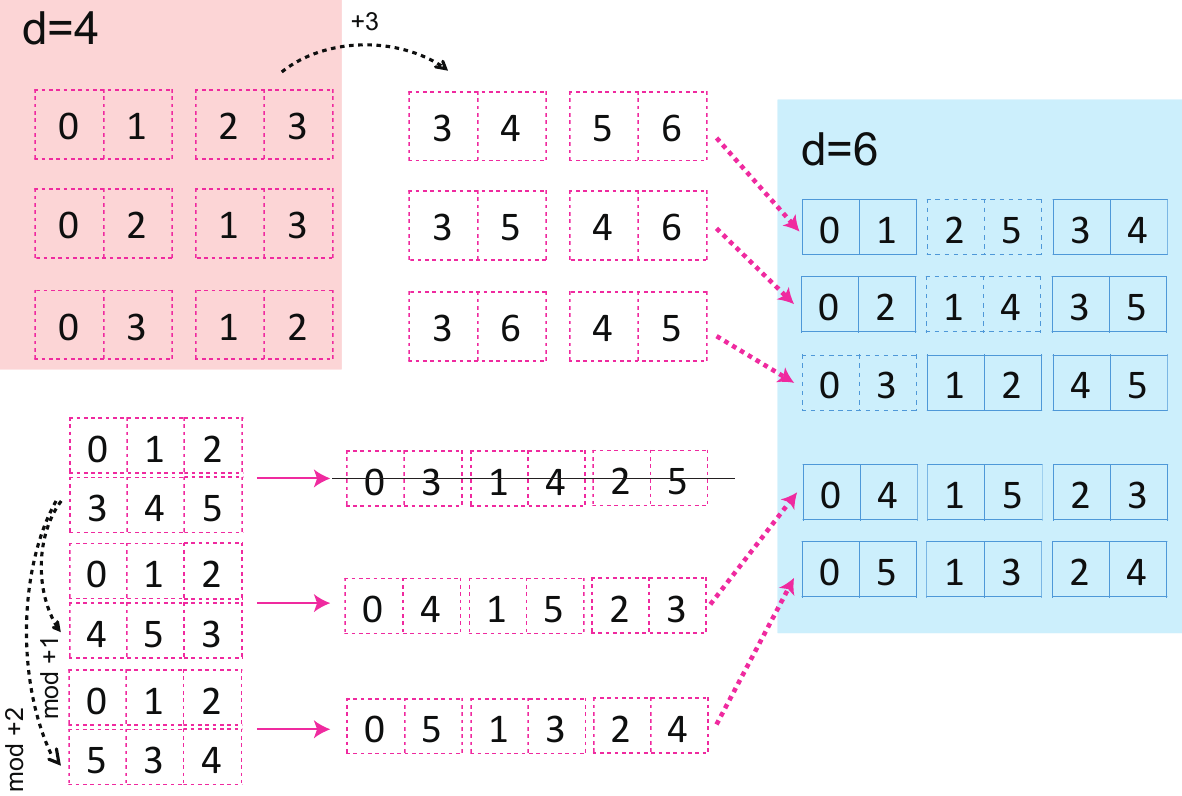}
\caption{Construction of five partitions for \(d=6\). 
%Here, the neighbors of \(3 = d/2\) are \(2, 1, 0\). The first partition, \(\{(0, 1), (2, 5), (3, 4)\}\), is obtained by combining and modifying existing pairs.
}
\label{partitions6}
\end{figure}

\textbf{Conclusion:}  
We verify that all pairs \((j, k)\) with \(0 \leq j < k \leq d-1\) are included in one of the partitions:

1. \textit{Both \(j\) and \(k\) in the same half (\(0 \leq j, k < d/2\) or \(d/2 \leq j, k < d\)):}  
   These pairs are covered by the merged partitions \(\{T_t^{d/2} \cup (T_t^{d/2} + d/2)\}\).

2. \textit{\(j\) and \(k\) in different halves (\(0 \leq j < d/2, d/2 \leq k < d\)):}  
    For even \(d/2\), pairs are covered by crossed partitions.
    For odd \(d/2\), pairs are covered by crossed or modified partitions.

Thus, the lemma holds for all even \(d\). \(\qed\)
\end{proof}

\begin{lemma}\label{lemma2}
    For any odd dimension \(d\), there exist \(d\) partitions \(\{T_1^d, T_2^d, \ldots, T_d^d\}\) such that any tuple \((j, k)\) with \(0 \leq j < k \leq d - 1\) is included in exactly one of these partitions.
\end{lemma}

\begin{proof}
We prove this by leveraging the result of Lemma \ref{lemma1}, which provides \(d+1\) partitions for the even dimension \(d+1\).

First, for the even dimension \(d+1\), each partition \(\{T_t^{d+1} : t = 1, \ldots, d\}\) includes one pair of the form \((C_t, d)\), where \(C_t\) is a neighboring number of \(d\). These pairs ensure that all tuples in the range \(0 \leq j < k \leq d\) are covered exactly once, while no duplicate elements appear within a partition.

Next, we construct the \(d\) partitions for the odd dimension \(d\) by modifying the partitions from Lemma \ref{lemma1}: 
For each partition \(T_t^{d+1}\), remove the tuple \((C_t, d)\), which contains the element \(d\). Replace \((C_t, d)\) with the singleton \(\{C_t\}\), ensuring that each element \(c_t\) in \(T_t^{d+1}\) is retained without the element \(d\). This modification ensures that all tuples are properly accounted for in the new partitions.

Finally, after these modifications, the resulting partitions \(\{T_t^d : t = 1, \ldots, d\}\) satisfy the requirements for dimension \(d\). Every tuple \((j, k)\) with \(0 \leq j < k \leq d - 1\) is included in exactly one partition. Additionally, the single elements \(\{C_t\}_{t=1}^d\) cover all integers in the set \(\{0, 1, \ldots, d-1\}\). This follows because the original set of pairs \(\{(C_t, d)\}\) in \(d+1\) partitions corresponds to \(\{(k, d)\}_{k=0}^{d-1}\), and replacing \(d\) with \(\{C_t\}\) maintains full coverage. This explains why the computational basis is not required in the DDB construction for odd dimensions \(d\). 

Thus, the construction guarantees \(d\) partitions for any odd dimension \(d\), completing the proof.
\end{proof}

\subsection{Algorithm for Arbitrary Dimension Based on the Lemmas}

\begin{algorithm}[H]
  \caption{Construct Minimal Partitions}
  \label{alg1}
  \begin{algorithmic}[1]
    \Require
      Positive integer \(d\)
    \Ensure
      \(d-1\) (even \(d\)) or \(d\) (odd \(d\)) partitions such that every pair \((j, k)\), \(j < k\), is included in exactly one partition
    \Statex
    \Function{ConstructPartitions}{$d$}
      \State Construct a sequence of even numbers \(\{b_l\}\) where:
      \[
      b_1 = f(d), \quad b_2 = f(b_1/2), \quad \dots, \quad b_L = 2,
      \]
      with \(f(x) = \begin{cases} 
        x & \text{if } x \text{ is even}, \\
        x + 1 & \text{if } x \text{ is odd}.
      \end{cases}\)
      
      \For{\(l = L \to 1\)}
        \If{\(l = L\)}
          \State Initialize \(\{T_1^2 = \{(0,1)\}\}\) for \(b_L = 2\)
        \Else
          \State \Call{SubProcedure}{$b_l, b_{l+1}$}
        \EndIf
      \EndFor
      
      \If{\(d\) is odd}
        \State Replace \((C_t, d)\) in partitions with \(C_t\) (single elements)
      \EndIf
      
      \State \Return All partitions for \(d\)
    \EndFunction
    \Statex
    
    \Algphase{SubProcedure: Iterative Construction of Partitions}
    \Function{SubProcedure}{$b_l, b_{l+1}$}
      \If{\(b_l/2\) is even}
        \State Relabel \(\{T_1^{b_{l+1}}, \dots, T_{b_{l+1}-1}^{b_{l+1}}\}\)
        \For{\(t = 1 \to b_l/2 - 1\)}
          \State Define \(T_t^{b_l} \gets T_t^{b_{l+1}} \cup (T_t^{b_{l+1}} + b_l/2)\)
        \EndFor
        \For{\(t = b_l/2 \to b_l - 1\)}
          \State Define \(T_t^{b_l} \gets \{(j, b_l/2 + [(j + t) \bmod b_l/2]) : j \in [0, b_l/2 - 1]\}\)
        \EndFor
      \Else
        \State Relabel \(\{T_1^{b_{l+1}}, \dots, T_{b_{l+1}}^{b_{l+1}}\}\)
        \For{\(t = 1 \to b_l/2\)}
          \State Define \(T_t^{b_l} \gets T_t^{b_{l+1}} \cup (T_t^{b_{l+1}} + b_l/2) - (j_t, b_l/2) - (b_l/2 + j_t, b_l) \cup (j_t, b_l/2 + j_t)\),
          \Statex \hspace{\algorithmicindent} where \((j_t, b_l/2) \in T_t^{b_{l+1}}\) and \(j_t\) is the neighbor of \(b_l/2\)
        \EndFor
        \For{\(t = b_l/2 + 1 \to b_l - 1\)}
          \State Define \(T_t^{b_l} \gets \{(j, b_l/2 + [(j + t) \bmod b_l/2]) : j \in [0, b_l/2 - 1]\}\)
        \EndFor
      \EndIf
    \EndFunction
  \end{algorithmic}
\end{algorithm}

\subsection{Examples of Minimal Partitions and Corresponding DDBs}
The deterministic algorithm can generate the minimal partitions required for any dimension \(d\). We will first present the partitions for 1-qubit, 2-qubit, and 3-qubit systems, which have dimensions of 2, 4, and 8, respectively. Subsequently, we will provide the corresponding DDBs for each partition. Then, we will demonstrate the partition and DDB construction for the odd dimension \(d = 7\).

\paragraph{Case \(d = 2\) (1-qubit)}
For \(d = 2\), there is only one tuple \((0, 1)\) with \(0 \leq j < k \leq 1\). Thus, the single partition is:
\begin{equation}
\mathbb{T}^2 = \{T_1^2\}, \quad T_1^2 = \{(0, 1)\}.
\end{equation}
The three informationally complete (IC) DDBs are:
\begin{equation}\label{1-qubit}
\begin{aligned}
  \mathcal{B}_0^2 &= \{|0\rangle, |1\rangle\}, \\
  \mathcal{B}_1^2 &= \left\{ \frac{|0\rangle \pm |1\rangle}{\sqrt{2}} \right\}, \quad
  \mathcal{C}_1^2 = \left\{ \frac{|0\rangle \pm i|1\rangle}{\sqrt{2}} \right\}.
\end{aligned}
\end{equation}

\paragraph{Case \(d = 4\) (2-qubit)}
For \(d = 4\), the three partitions are:
\begin{equation}
\mathbb{T}^4 = \{T_1^4, T_2^4, T_3^4\}, \quad \text{where}
\end{equation}
\begin{equation}\label{2-qubit-partitions}
\begin{aligned}
  T_1^4 &= \{(0, 1), (2, 3)\}, \\
  T_2^4 &= \{(0, 2), (1, 3)\}, \quad
  T_3^4 = \{(0, 3), (1, 2)\}.
\end{aligned}
\end{equation}
The seven IC DDBs are:
\begin{equation}\label{2-qubit}
\begin{aligned}
  \mathcal{B}_0^4 &= \{|0\rangle, |1\rangle, |2\rangle, |3\rangle\}, \\
  \mathcal{B}_1^4 &= \left\{ \frac{|0\rangle \pm |1\rangle}{\sqrt{2}}, \frac{|2\rangle \pm |3\rangle}{\sqrt{2}} \right\}, \quad
  \mathcal{C}_1^4 = \left\{ \frac{|0\rangle \pm i|1\rangle}{\sqrt{2}}, \frac{|2\rangle \pm i|3\rangle}{\sqrt{2}} \right\}, \\
  \mathcal{B}_2^4 &= \left\{ \frac{|0\rangle \pm |2\rangle}{\sqrt{2}}, \frac{|1\rangle \pm |3\rangle}{\sqrt{2}} \right\}, \quad 
  \mathcal{C}_2^4 = \left\{ \frac{|0\rangle \pm i|2\rangle}{\sqrt{2}}, \frac{|1\rangle \pm i|3\rangle}{\sqrt{2}} \right\}, \\
  \mathcal{B}_3^4 &= \left\{ \frac{|0\rangle \pm |3\rangle}{\sqrt{2}}, \frac{|1\rangle \pm |2\rangle}{\sqrt{2}} \right\}, \quad 
  \mathcal{C}_3^4 = \left\{ \frac{|0\rangle \pm i|3\rangle}{\sqrt{2}}, \frac{|1\rangle \pm i|2\rangle}{\sqrt{2}} \right\}.
\end{aligned}
\end{equation}

\paragraph{Case \(d = 8\) (3-qubit)}
For \(d = 8\), the seven partitions are:
\begin{equation}\label{3-qubit-partitions}
\begin{aligned}
  T_1^8 &= \{(0, 1), (2, 3), (4, 5), (6, 7)\}, \\
  T_2^8 &= \{(0, 2), (1, 3), (4, 6), (5, 7)\}, \\
  T_3^8 &= \{(0, 3), (1, 2), (4, 7), (5, 6)\}, \\
  T_4^8 &= \{(0, 4), (1, 5), (2, 6), (3, 7)\}, \\
  T_5^8 &= \{(0, 5), (1, 6), (2, 7), (3, 4)\}, \\
  T_6^8 &= \{(0, 6), (1, 7), (2, 4), (3, 5)\}, \\
  T_7^8 &= \{(0, 7), (1, 4), (2, 5), (3, 6)\}.
\end{aligned}
\end{equation}
The fifteen IC DDBs are constructed following similar principles as for \(d = 4\), with analogous forms of \(\mathcal{B}_t^8\) and \(\mathcal{C}_t^8\).

\paragraph{Case \(d = 7\) (Odd Dimension)}
For \(d = 7\), the seven partitions are constructed as follows:
\begin{equation}\label{7-partitions}
\begin{aligned}
  T_1^7 &= \{(0, 1), (2, 3), (4, 5), 6\}, \quad
  T_2^7 = \{(0, 2), (1, 3), (4, 6), 5\}, \\
  T_3^7 &= \{(0, 3), (1, 2), 4, (5, 6)\}, \quad
  T_4^7 = \{(0, 4), (1, 5), (2, 6), 3\}, \\
  T_5^7 &= \{(0, 5), (1, 6), 2, (3, 4)\}, \quad
  T_6^7 = \{(0, 6), 1, (2, 4), (3, 5)\}, \\
  T_7^7 &= \{0, (1, 4), (2, 5), (3, 6)\}.
\end{aligned}
\end{equation}
The fourteen IC DDBs are similarly defined, with explicit forms for \(\mathcal{B}_t^7\) and \(\mathcal{C}_t^7\).

The computational basis \(\mathcal{B}_0^7\) is excluded, as all elements \(\{|0\rangle, \cdots, |6\rangle\}\) are already covered by other DDBs.

\section{Numerical experiments on dimension six}

While it is known that the PMs onto \(d+1\) MUBs are minimal and optimal QST strategy for a \(d\)-dimensional system, their construction for each dimension \(d\) is still an open question. The first dimension for which \(d+1\) MUBs have not been constructed is 6, corresponding to a qubit-qutrit system, \(\mathcal{H}_2\otimes \mathcal{H}_3\).

For \(d=6\), there are 5 partitions in \(\mathbb{T}^6\),
\begin{eqnarray}
&&T^6_1=\{(0,1),(2,5),(3,4)\}, \quad T^6_2=\{(0,2),(1,4),(3,5)\}, \quad T^6_3=\{(0,3),(1,2),(4,5)\}, \nonumber \\
&&T^6_4=\{(0,4),(1,5),(2,3)\}, \quad T^6_5=\{(0,5),(1,3),(2,4)\}.
\end{eqnarray}

The corresponding 11 IC DDBs are denoted as \(\{\mathcal{B}_0^6,\mathcal{B}_1^6,\ldots,\mathcal{B}_5^6,\mathcal{C}_1^6,\ldots,\mathcal{C}_5^6\}\),
\begin{equation}
\begin{aligned}
  \mathcal{B}_0^6 &= \{|0\rangle, \dots, |5\rangle\}, \\
  \mathcal{B}_1^6 &= \left\{ |\phi_{01}^{\pm}\rangle, |\phi_{25}^{\pm}\rangle, |\phi_{34}^{\pm}\rangle \right\}, \quad
  \mathcal{C}_1^6 = \left\{ |\psi_{01}^{\pm}\rangle, |\psi_{25}^{\pm}\rangle, |\psi_{34}^{\pm}\rangle \right\}, \\
  \mathcal{B}_2^6 &= \left\{ |\phi_{02}^{\pm}\rangle, |\phi_{14}^{\pm}\rangle, |\phi_{35}^{\pm}\rangle \right\}, \quad
  \mathcal{C}_2^6 = \left\{ |\psi_{02}^{\pm}\rangle, |\psi_{14}^{\pm}\rangle, |\psi_{35}^{\pm}\rangle \right\}, \\
  \mathcal{B}_3^6 &= \left\{ |\phi_{03}^{\pm}\rangle, |\phi_{12}^{\pm}\rangle, |\phi_{45}^{\pm}\rangle \right\}, \quad
  \mathcal{C}_3^6 = \left\{ |\psi_{03}^{\pm}\rangle, |\psi_{12}^{\pm}\rangle, |\psi_{45}^{\pm}\rangle \right\}, \\
  \mathcal{B}_4^6 &= \left\{ |\phi_{04}^{\pm}\rangle, |\phi_{15}^{\pm}\rangle, |\phi_{23}^{\pm}\rangle \right\}, \quad
  \mathcal{C}_4^6 = \left\{ |\psi_{04}^{\pm}\rangle, |\psi_{15}^{\pm}\rangle, |\psi_{23}^{\pm}\rangle \right\}, \\
  \mathcal{B}_5^6 &= \left\{ |\phi_{05}^{\pm}\rangle, |\phi_{13}^{\pm}\rangle, |\phi_{24}^{\pm}\rangle \right\}, \quad
  \mathcal{C}_5^6 = \left\{ |\psi_{05}^{\pm}\rangle, |\psi_{13}^{\pm}\rangle, |\psi_{24}^{\pm}\rangle \right\}.
\end{aligned}
\end{equation}

We tested our proposal numerically in a 6-dimensional system. We examined four quantum state types: (a) the maximally mixed state \(I/6\), (b) a balanced state \(\frac{1}{6}\sum_{k,j=0}^5|k\rangle\langle j|\), (c) a separable state, and (d) an entangled state. These states are Hermitian, semi-definite, and unit trace density matrices. States (a) and (b) were directly generated in our simulation, while states (c) and (d) were prepared using local unitary transformations \(U_2 \in \mathcal{H}_2\) and \(U_3 \in \mathcal{H}_3\), distributed uniformly according to the Haar measure. As \(U_2\) and \(U_3\) are local, the entanglement of the resulting state remains unchanged, which can be verified using the Peres-Horodecki criterion. We prepared states (c) and (d) as \(U_2\otimes U_3 \ket{\phi}\), where \(\ket{\phi}\) represents \(\ket{0}\) or \(\ket{1}+\ket{2}+\ket{3}+\ket{5}\), respectively.

For the reconstruction of an unknown density matrix \(\rho\), we utilize two methods. The first method is based on semi-definite programming, where \(\tilde{\rho}\) represents the estimated form of \(\rho\) and is obtained through a parameterized matrix \(X \geq 0\). The mathematical model is formulated as follows:
\begin{eqnarray}
  \tilde{\rho} = \mbox{arg}\mathop{\mbox{min}}_{X}\sum_{i=1}^{66}  \|( \mbox{tr}(X E_i)-p_i)\|,
\end{eqnarray}
where \(\| \cdot \|\) is a norm function, \(E_i\) are from the bases of (\(\mathcal{B}_0,\ldots,\mathcal{B}_5,\mathcal{C}_1,\ldots,\mathcal{C}_5\)), and \(p_i\) are the measured probabilities on \(E_i\) by the unknown density matrix \(\rho\).

The second approach is direct reconstruction. This method utilizes a total of 36 probabilities, which is half of the available data.

Consequently, numerical experiments were conducted 20 times for all tested states. The Monte Carlo method was utilized to simulate \(p_i\) with \(100\times2^{\mbox{Num}}\) shots. The results of these numerical experiments are illustrated in SFig.(\ref{simu2}), where the infidelities are represented by the Frobenius distance:
\begin{eqnarray}
  F_{f} = \sqrt{\mbox{trace}((\rho-\tilde{\rho})\cdot (\rho-\tilde{\rho})^{\dagger})}.
\end{eqnarray}
Error bars come from standard deviations of 20 repetitions of the simulation.

\begin{figure}[!htb]
  \begin{center}
    \includegraphics[width=1\textwidth]{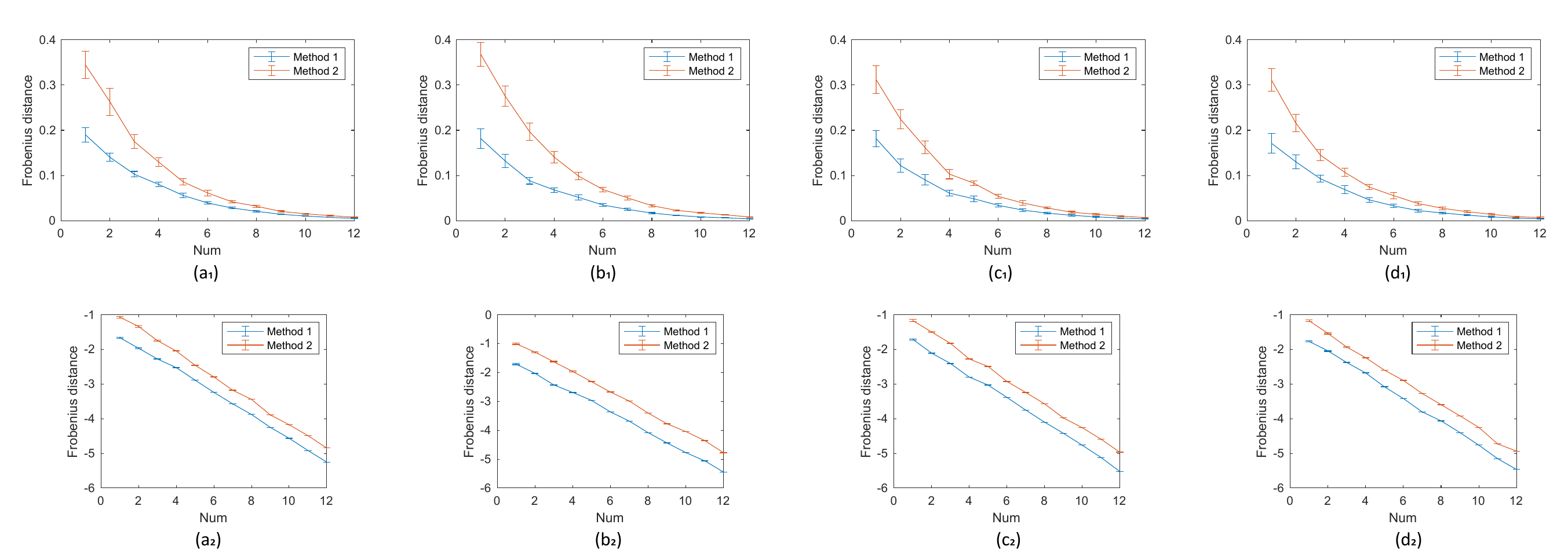}
    \caption{Numerical experiment for qubit-qutrit system. (a) Maximally mixed state, (b) balanced state, (c) separable state, and (d) entangled state are tested. The infidelities are calculated by Frobenius distance. Error bars come from the standard deviation of 20 repetitions of the simulation. Method 1 is via semi-definite programming, and method 2 is via direct estimation. The second line gives the logarithmic scale form.}
    \label{simu2}
  \end{center}
\end{figure}

\section{Rigorous proofs of Result 2}

\begin{lemma}\label{lemma1_proof}
    When \( d = 2^n \), the pair \((j, k)\) with \(0 \le j < k \le 2^n - 1\) and \(|j - k| \le r\) can be found in at most \( O(r (n - \log r)) \) partitions for minimal DDBs.
\end{lemma} 

\textbf{Proof.} Using \( n \) iterations, \( 2^n - 1 \) partitions are constructed such that each pair \((j, k)\) with \(0 \le j < k \le 2^n - 1\) is covered. When \( r \ll 2^n \), we can always find \( m \) such that \( 2^{m-1} < r \le 2^m \).

At iteration \( t = m \) for dimension \( 2^m \), a total of \( 2^m - 1 \) partitions have been constructed by Result~\ref{result1}. Thus, at most \( 2^m - 1 \) partitions contain \((j, k)\) for \(0 \le j < k \le 2^m - 1\). At iteration \( t = m + 1 \), new partitions are constructed by Eq. (\ref{merge1}) and Eq. (\ref{insect1}), resulting in at most \( 2^m - 1 \) and \( r \) partitions that satisfy \(|j - k| \le r\), \(0 \le j < k \le 2^{m+1} - 1\). 
This means that each iteration at most $r$ new partitions are added for the target pairs. 
Therefore, from \( t = m \) to \( t = \log d = n \), the number of relevant partitions is less than 
\[
2^m - 1 + r(\log d - m) < r(\log d + 2 - m) < r \log(4d/r). 
\] 
\qed

\begin{lemma}\label{lemma2_proof}
    For any general dimension \( d \), the pair \((j, k)\) with \(0 \le j < k \le d - 1\) and \(|j - k| \le r\) can be found in at most \( O(r \log \frac{d}{r})) \) partitions for minimal DDBs.
\end{lemma}

\textbf{Proof.} Based on the construction in Algorithm~\ref{alg1}, we iteratively construct $b_{k}-1$ partitions for dimension $b_k$, where $k=L,L-1,\cdots,1$. 
Here \(b_L = 2\) and \(b_1 = f(d)\). 

Firstly, we prove that the number of iterations is exactly \(L = \lceil \log d \rceil\). 
For example, consider \(d = 100\). We should construct the partitions iteratively for the dimensions:
\[
b_7 = 2, \; b_6 = 4, \; b_5 = 8, \; b_4 = 14, \; b_3 = 26, \; b_2 = 50, \; b_1 = 100.
\] 
With \(\lceil \log 100 \rceil = 7\) iterations, we obtain the partitions for \(d = 100\). 

For any general \( d \), we observe that \( b_{k-1} = 2b_k \) or \( b_{k-1} = 2b_k - 2 \), leading to the inequality \( b_{k-1} \leq 2b_k \). If the iteration count is \( \lceil \log d \rceil - 1 \), the maximum value of \( b_1 \) is \( 2^{\lceil \log d \rceil - 1} \). 
However, \( 2^{\lceil \log d \rceil - 1} < d \). This results in a contradiction since \( b_1 \) should equal \( d \) or \( d + 1 \). On the other hand, it is easy to observe that if we start with dimension \( b_1=f(d) \) and iterate \( L = \lceil \log d \rceil \) times, the final value can always be reduced to 2. This is because \( 2^{L-1} < b_1 \le 2^L \), and generally \( 2^{L-k} < b_k \le 2^{L-k+1} \) for \( k=2,\cdots,L \). 

Similar to the case when \( d=2^n \), there is always a number \( m \) for the rank $r$ such that \( 2^{m-1} < r \le 2^m \). 
At the \( m \)-iteration, we should construct the partitions for dimension $b_{L-{m-1}}$. 
Thus, there are at most \( b_{L-m+1}-1\le 2^m-1 \) partitions for dimension \( b_{L+1-m} \). 
Repeating the procedure until the iteration \( L \) for dimension \( b_1 \), the new constructions are based on Eq. (\ref{merge2}) and Eq. (\ref{insect2}). The number of partitions containing the required elements is less than 
\[
2^m-1+r( L-m)< r(1+\lceil \log d \rceil -m)<r\log (4d/r), 
\]
where \( -m \le -\log r \). Then the number of DDBs required is \(O(r \log(d/r))\).  \qed 

Each partition corresponds to two eigenbases. Together with \(\mathcal{B}_0\), \(O(r \log \frac{d}{r})\) DDBs can reconstruct the elements \(\{\rho_{jk} : j,k \in C\}\).

\section{Circuits analysis} \label{sec:circuit}

We may as well label the $2^{n+1}-1$ DDBs on $n$-qubit systems as follows: 
\begin{eqnarray}
	\{\mathcal{B}_0^{2^n},  {\mathcal{B}_j^{2^n}, \mathcal{C}_j^{2^n}:j=1,\cdots,2^n-1}\}.
\end{eqnarray}
The computational basis is $\mathcal{B}_0^{2^n}=\{|0\rangle,\cdots,|2^n-1\rangle\}$. 
The basis $\mathcal{B}_j^{2^n}$ and $\mathcal{C}_j^{2^n}$ are dually designed for the same partition. Here we consider the $2^{n+1}-2$ circuits to transform the computational basis to nontrivial DDBs $\{\mathcal{B}_j^{2^n}, \mathcal{C}_j^{2^n}:j=1,\cdots,2^n-1\}$.

\textit{1-qubit:}  The DDBs $\mathcal{B}_1^2$ and $\mathcal{C}_1^2$ are in Eq. (\ref{1-qubit}). The DDB circuits to map the computational basis into them are shown in SFig.(\ref{1qubit}). 
\begin{figure}[!htb]
	\[
	\Qcircuit @C=0.8em @R=0.8em {
		 & \gate{H} & \qw  }
  ~~~~ 
  \Qcircuit @C=0.8em @R=0.8em {
		 & \gate{H} & \gate{S} & \qw  }
	\]
	\caption{Circuits to obtain $\mathcal{B}_1^2$ and $\mathcal{C}_1^2$. In order to perform PMs onto $\mathcal{B}_1^2$ and $\mathcal{C}_1^2$, we should apply $H^{\dag}=H$ and $\tilde{H}^{\dag}=(HS)^{\dag}=S^{\dag}H$ followed by computation basis measurement. They are exactly the Pauli measurement $X$ and Pauli measurement $Y$.}
	\label{1qubit}
\end{figure}
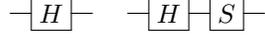 

\textit{2-qubit:} The six nontrivial DDBs are in Eq. (\ref{2-qubit}). With the binary form, the three DDBs (without coefficient $i$) for the three partitions are the following 
\begin{equation}
\begin{aligned}
	\mathcal{B}_1^4 &= \left\{ \frac{|00\rangle \pm |01\rangle}{\sqrt{2}}, \frac{|10\rangle \pm |11\rangle}{\sqrt{2}} \right\}, \\
	\mathcal{B}_2^4 &= \left\{ \frac{|00\rangle \pm |10\rangle}{\sqrt{2}}, \frac{|01\rangle \pm |11\rangle}{\sqrt{2}} \right\}, \\
	\mathcal{B}_3^4 &= \left\{ \frac{|00\rangle \pm |11\rangle}{\sqrt{2}}, \frac{|01\rangle \pm |10\rangle}{\sqrt{2}} \right\}.
\end{aligned}
\end{equation}
The corresponding circuits are depicted in Fig. (\ref{2-qubit1}):  
\begin{figure}[!htb]
  \[
  \Qcircuit @C=0.8em @R=0.8em {
\lstick{q_1}     & \qw & \qw \\
\lstick{q_2}     & \gate{H} & \qw }
 ~~~~~~~~~~
  \Qcircuit @C=0.8em @R=0.8em {
\lstick{q_1}  	  & \gate{H} &  \qw  \\
\lstick{q_2}	 & \qw 		& \qw  }
 ~~~~~~~~~~
  \Qcircuit @C=0.8em @R=0.8em {
\lstick{q_1}    	  & \gate{H} &\ctrl{1} &\qw    \\
\lstick{q_2}   	 		&\qw &\targ & \qw  }
 \]
  \caption{Circuits to obtain $\mathcal{B}_1^4,\mathcal{B}_2^4,\mathcal{B}_3^4$. It is easy to verify that they map the computational basis $\{|00\rangle,|01\rangle,|10\rangle,|11\rangle\}$ into the designed ones. By changing the gate $H$ of the circuits into $\tilde{H}$, we will obtain the dual DDBs $\mathcal{C}_1^4,\mathcal{C}_2^4,\mathcal{C}_3^4$.}
  \label{2-qubit1}
\end{figure}
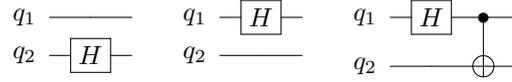

\textit{$n$-qubit:} 
Denote the \(2^n-1\) unitary operations corresponding to the nontrivial DDBs \(\mathcal{B}_t^{2^n}\) as \(\{U_t^{2^n} : t = 1, \cdots, 2^n - 1\}\). The dual DDBs \(\{\mathcal{C}_t^{2^n}\}\) are obtained by adding a global phase factor \(i\) to the basis states. Replacing the \(H\) gate in the circuits with \(\tilde{H}=HS\) yields the circuits for these dual DDBs \(\mathcal{C}_t^{2^n}\). 
At the last iteration of \((n-1)\)-qubit case, the unitary operations  \(\{U_t^{2^{n-1}} : t = 1, \cdots, 2^{n-1} - 1\}\) map the computational basis \(\mathcal{B}_0^{2^{n-1}}\) into \(\mathcal{B}_t^{2^{n-1}}\). 
%The basis \(\mathcal{B}_t^{2^{n-1}}\) is constructed for the partition \(T_t^{2^{n-1}}\).

For the \(n\)-qubit case, we have \(T_t^{2^{n}} = T_t^{2^{n-1}} \cup (T_t^{2^{n-1}} + 2^{n-1})\) for \(t = 0, \cdots, 2^{n-1} - 1\), and \(T_t^{2^{n}} = \{(j, 2^{n-1} + [(j + t) \bmod 2^{n-1}]) : 0 \le j \le 2^{n-1} - 1\}\) for \(t = 2^{n-1}, \cdots, 2^{n} - 1\).

Thus, when \(t = 0, \cdots, 2^{n-1} - 1\), we have \(U_t^{2^n} = I \otimes U_t^{2^{n-1}}\). This is because the basis states of \(\mathcal{B}_t^{2^{n-1}}\) are of the form \(|k_1\rangle \pm |k_2\rangle\). The partition \(T_t^{2^{n}}\) is iteratively constructed by \(T_t^{2^{n-1}} \cup (T_t^{2^{n-1}} + 2^{n-1})\), where \(t = 0, \cdots, 2^{n-1} - 1\). Therefore, the basis states of \(\mathcal{B}_k^{2^{n}}\) are in the form \(|0\rangle(|k_1\rangle \pm |k_2\rangle)\) or \(|1\rangle(|k_1\rangle \pm |k_2\rangle)\).

When \(t = 2^{n-1}\), we have \(U_{2^{n-1}}^{2^n} = H \otimes I^{\otimes n-1}\). This follows because the partition \(T_{2^{n-1}}^{2^n}\) consists of \(\{(0, 2^{n-1}), (1, 2^{n-1} + 1), \cdots, (2^{n-1} - 1, 2^{n} - 1)\}\). These numbers can be expressed in binary form, and the corresponding basis states are given by \(\{(|0\rangle \pm |1\rangle) \otimes |j_2, \cdots, j_n\rangle : j_2, \cdots, j_n = 0, 1\}\). Hence, \(U_{2^{n-1}}^{2^n} = H \otimes I^{\otimes n-1}\).

When \(t = 2^{n-1} + 1, \cdots, 2^n - 1\), we can express \(t\) as \(2^{n-1} + j\). Then, \(U_{k}^{2^n} =  [|0\rangle\langle 0|\otimes I + |1\rangle\langle 1|\otimes (\mathcal{V}_{n-1})^j] \cdot [H \otimes I^{\otimes (n-1)}]\), where 
\begin{equation}
    \mathcal{V}_{n-1} =\mathcal{U}^{\dag}_{n-1}= \sum_{m=0}^{2^{n-1} - 1} |m + 1 \bmod 2^{n-1}\rangle\langle m|.
\end{equation}

This result is derived as follows. The partition \(T_{t}^{2^n}\) consists of \(\{(0, 2^{n-1} + j), (1, 2^{n-1} + (1 + j \bmod 2^{n-1})), \cdots, (2^{n-1} - 1, 2^{n-1} + (2^{n-1} - 1 + j \bmod 2^{n-1}))\}\). Expressing these numbers in binary form, the corresponding basis states are given by \(\{|0\rangle \otimes |j_2, \cdots, j_n\rangle + |1\rangle \otimes |[(j_2, \cdots, j_n) + j] \bmod 2^{n-1}\rangle : j_2, \cdots, j_n = 0, 1\}\). This basis can be obtained by applying the conditional shift operation \(|0\rangle\langle 0|\otimes I + |1\rangle\langle 1|\otimes (\mathcal{U}^{\dag}_{n-1})^j\) on the basis of \(\mathcal{B}_{2^{n-1}}^{2^n}\).

\begin{figure}[!htb]
  \[
  \Qcircuit @C=1em @R=1em {
\lstick{q_1}   & \qw & \qw & \qw    &&&  &\gate{H} &\qw  &&& &\qw 	  & \gate{H} &\ctrl{1} &\qw&\qw\\
\lstick{q_2}   & \qw & \multigate{3}{U_t^{2^{n-1}}} & \qw &&&  &\qw &\qw  &&& & \qw & \qw & \multigate{3}{\mathcal{V}_{n-1}^{ j}} & \qw&\qw\\
\lstick{\vdots}   &   &   &   &&& & \vdots  &  &&& &  & \vdots   &   & &\\
\lstick{q_{n-1}}   & \qw & \ghost{U_t^{2^{n-1}}} & \qw &&& &\qw &\qw  &&& & \qw & \qw &\ghost{\mathcal{V}_{n-1}^{ j}}& \qw&\qw\\
\lstick{q_n}   & \qw & \ghost{U_t^{2^{n-1}}} & \qw &&& &\qw &\qw  &&& & \qw & \qw& \ghost{\mathcal{V}_{n-1}^{ j}} & \qw&\qw}
 \]
  \caption{Circuits to obtain all DDBs $\mathcal{B}_t^{2^n}$.  The left circuit represents $\{U_t^{2^n}:t=1,\cdots,2^{n-1}-1\}$. The middle circuit represent $U_{2^{n-1}}^{2^n}$. The right circuit represents $\{U_t^{2^n}:t=2^{n-1}+1,\cdots,2^{n}-1\}$, where $j=t-2^{n-1}$. Similarly, by changing $H$ gate into $\tilde{H}$, we will obtain the dual DDBs for $\mathcal{C}_t^{2^n}$.}
  \label{N-qubit1}
\end{figure}
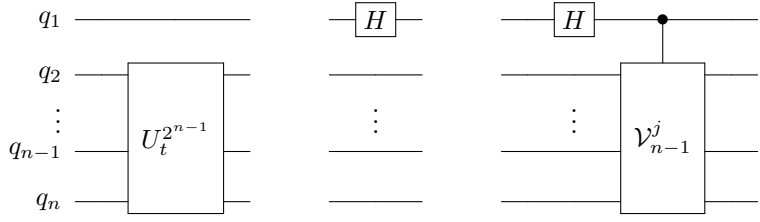

\textbf{Remark:} The circuits shown in the middle and right panels of Fig.~(\ref{N-qubit1}) are newly constructed during the \(n\)-th iteration. At each \(k\)-th iteration (\(k=2, \ldots, n-1\)), a total of \(2^k\) circuits are depicted in Fig.~(\ref{k-qubit1}). Specifically, the left panel represents the circuits from the \((k-1)\)-th iteration, the middle panel corresponds to the circuits for \(\mathcal{B}_{2^{k-1}}^{2^k}\), and the right panel illustrates the crossed partitions introduced during the \(k\)-th iteration. 

From iteration 1 to \(n\), the total number of constructed circuits is given by \(2^0 + 2^1 + \ldots + 2^{n-1} = 2^n - 1\). 
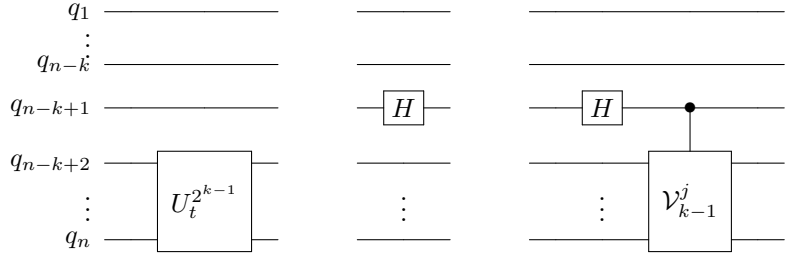
\begin{figure}[!htb]
  \[
  \Qcircuit @C=1em @R=1em {
\lstick{q_1}   & \qw & \qw & \qw    &&&  &\qw &\qw  &&& &\qw 	  & \qw &\qw &\qw&\qw\\
\lstick{\vdots}   &    &     &    &&&  &   &  &&& &   &    &   &  & \\
\lstick{q_{n-k}}   & \qw & \qw & \qw &&& &\qw &\qw  &&& &\qw & \qw & \qw & \qw&\qw\\
\lstick{q_{n-k+1}}   & \qw & \qw & \qw &&& &\gate{H} &\qw  &&& &\qw & \gate{H} & \ctrl{1} & \qw&\qw\\
\lstick{q_{n-k+2}}  & \qw & \multigate{2}{U_t^{2^{k-1}}}  & \qw &&& &\qw &\qw  &&& & \qw & \qw &\multigate{2}{\mathcal{V}_{k-1}^{ j}}& \qw&\qw\\
\lstick{\vdots}  &   &    &   &&& &\vdots &   &&& &   & \vdots  & &  & \\
\lstick{q_n}  & \qw & \ghost{U_t^{2^{k-1}}}  & \qw &&& &\qw &\qw  &&& & \qw & \qw& \ghost{\mathcal{V}_{k-1}^{ j}} & \qw&\qw}
 \]
  \caption{Circuits to obtain $\mathcal{B}_t^{2^k}$. The operations involve leaving the first \(n-k\) qubits unchanged by applying the identity gate \(I\). The left circuit represents the set \(\{U_t^{2^k}=I\otimes U_t^{2^{k-1}}  : t = 1, \dots, 2^{k-1} - 1\}\). The middle circuit corresponds to \(U_{2^{k-1}}^{2^k}\). The right circuit represents the set \(\{U_t^{2^k} : t = 2^{k-1} + 1, \dots, 2^k - 1\}\), where \(j = t - 2^{k-1}\). Similarly, by replacing the \(H\) gate with \(\tilde{H}\),  can obtain the dual DDBs for \(\mathcal{C}_t^{2^k}\).}
  \label{k-qubit1}
\end{figure}

\subsection{Circuits decomposition of permutational operation}

As previously mentioned, the PM onto the eigenbasis \(\{U|k\rangle : k = 0, \cdots, 2^n - 1\}\) can be implemented by first applying \(U^{\dagger}\), followed by a PM onto the computational basis \(\mathcal{B}_0^{2^n}\). 

The operation \(H\) (or \(\tilde{H}^{\dag}\)) followed by a single-qubit computational basis measurement (Pauli \(Z\) measurement) is equivalent to performing a Pauli \(X\) (or \(Y\)) measurement. Therefore, it suffices to decompose the conjugate transpose of the permutation operation depicted in Fig. (\ref{k-qubit1}):  
\begin{equation}\label{Pt}
   T_{j,k} = |0\rangle\langle 0| \otimes I + |1\rangle\langle 1| \otimes (\mathcal{U}_{k-1})^j,
\end{equation}
where \(\mathcal{U}_k = \sum_{m=0}^{2^k-1} |m-1 \bmod 2^k\rangle\langle m|\), with \(k = 2, \cdots, n-1\) and \(j = 1, \cdots, 2^{k-1} - 1\). These operations are the relabels of permutation operation $P_{j,k}$ in the main text.  

In the following, we demonstrate that all such permutation operations can be efficiently decomposed into a polynomial number of elementary gates.

\begin{Decomposition}\label{decom1}
For each \( j = 1, \dots, 2^{k}-1 \), even when \( j \) is exponentially large, \( (\mathcal{U}_{k})^j \) can be decomposed into a product of at most \( k \) specific unitary operations, namely \( (\mathcal{U}_{k})^{2^0}, (\mathcal{U}_{k})^{2^1}, \dots, (\mathcal{U}_{k})^{2^{k-1}} \).
\end{Decomposition}

\textbf{Analysis:} For each \( j \in \{1, \dots, 2^k-1\} \), we can represent \( j \) in binary form as \(\vec{j}= (j_0, j_1, \dots, j_{k-1}) \), where \( j_i \in \{0,1\} \). Specifically, \( j \) is expressed as:
\[
j = j_0 \times 2^0 + j_1 \times 2^1 + \cdots + j_{k-1} \times 2^{k-1}.
\]

Therefore, we have:
\[
(\mathcal{U}_{k})^j = [(\mathcal{U}_{k})^{2^0}]^{j_0} \times [(\mathcal{U}_{k})^{2^1}]^{j_1} \times \cdots \times [(\mathcal{U}_{k})^{2^{k-1}}]^{j_{k-1}}.
\]

The quantum circuit for this decomposition is illustrated in Fig.~(\ref{fig:ukj}). 

For example, if \( j_0 = 0 \), the term \([(\mathcal{U}_{k})^{2^0}]^{j_0}\) simplifies to the identity operation \( I \) and is omitted. Consequently, we only need to identify the nonzero elements in \(\{j_0, j_1, \dots, j_{k-1}\}\) and execute at most \( k \) operations from the set \(\{(\mathcal{U}_{k})^1, (\mathcal{U}_{k})^{2^1}, \dots, (\mathcal{U}_{k})^{2^{k-1}}\}\). Since these operations commute with one another, the execution order is flexible. \qed

\begin{figure}[h]
    \centering
    \[
    \Qcircuit @C=1em @R=1em {
    & \gate{(\mathcal{U}_{k})^j} & \qw &  & = & & &\qw & \gate{(\mathcal{U}_{k})^1} & \qw & \gate{(\mathcal{U}_{k})^{2^1}} & \qw & \cdots &   & \gate{(\mathcal{U}_{k})^{2^{k-1}}} & \qw    
    \gategroup{1}{9}{1}{9}{0.8em}{--}  \gategroup{1}{11}{1}{11}{0.8em}{--}  \gategroup{1}{15}{1}{15}{0.8em}{--}  \\
    &   &   & & &   & &   & j_0 &  & j_1 &   &   &   & j_{k-1} &    
}
\]
    \caption{Quantum circuit for \((\mathcal{U}_k)^j\). Dashed boxes indicate that the corresponding circuits are conditionally executed based on the binary coefficients \(j_0, j_1, \ldots, j_{k-1}\). For instance, if \(j_0 = 1\), the circuit for \((\mathcal{U}_k)^1\) is executed; otherwise, the identity operation \(I\) is applied.}
    \label{fig:ukj}
\end{figure}

\begin{Decomposition}\label{decom1}
For each \( j = 1, \dots, 2^{k}-1 \), even when \( j \) is exponentially large, \( (\mathcal{U}_{k})^j \) can be decomposed into a product of at most \( k \) specific unitary operations, namely \( (\mathcal{U}_{k})^{1}, (\mathcal{U}_{k})^{2^1}, \dots, (\mathcal{U}_{k})^{2^{k-1}} \).
\end{Decomposition}

\textbf{Analysis:} For each \( j \in \{1, \dots, 2^k-1\} \), we can represent \( j \) in binary form as \(\vec{j}= (j_0, j_1, \dots, j_{k-1}) \), where \( j_i \in \{0,1\} \). Specifically, \( j \) is expressed as:
\[
j = j_0 \times 2^0 + j_1 \times 2^1 + \cdots + j_{k-1} \times 2^{k-1}.
\]

Therefore, we have:
\[
(\mathcal{U}_{k})^j = [(\mathcal{U}_{k})^{2^0}]^{j_0} \times [(\mathcal{U}_{k})^{2^1}]^{j_1} \times \cdots \times [(\mathcal{U}_{k})^{2^{k-1}}]^{j_{k-1}}.
\]

The quantum circuit for this decomposition is illustrated in Fig.~(\ref{fig:ukj}). 

For example, if \( j_0 = 0 \), the term \([(\mathcal{U}_{k})^{2^0}]^{j_0}\) simplifies to the identity operation \( I \) and is omitted. Consequently, we only need to identify the nonzero elements in \(\{j_0, j_1, \dots, j_{k-1}\}\) and execute at most \( k \) operations from the set \(\{(\mathcal{U}_{k})^1, (\mathcal{U}_{k})^{2^1}, \dots, (\mathcal{U}_{k})^{2^{k-1}}\}\). Since these operations commute with one another, the execution order is flexible. \qed

\begin{figure}[h]
    \centering
    \[
    \Qcircuit @C=1em @R=1em {
    & \gate{(\mathcal{U}_{k})^j} & \qw &  & = & & &\qw & \gate{(\mathcal{U}_{k})^1} & \qw & \gate{(\mathcal{U}_{k})^{2^1}} & \qw & \cdots &   & \gate{(\mathcal{U}_{k})^{2^{k-1}}} & \qw    
    \gategroup{1}{9}{1}{9}{0.8em}{--}  \gategroup{1}{11}{1}{11}{0.8em}{--}  \gategroup{1}{15}{1}{15}{0.8em}{--}  \\
    &   &   & & &   & &   & j_0 &  & j_1 &   &   &   & j_{k-1} &    
}
\]
    \caption{Quantum circuit for \((\mathcal{U}_k)^j\). Dashed boxes indicate that the corresponding circuits are conditionally executed based on the binary coefficients \(j_0, j_1, \ldots, j_{k-1}\). For instance, if \(j_0 = 1\), the circuit for \((\mathcal{U}_k)^1\) is executed; otherwise, the identity operation \(I\) is applied.}
    \label{fig:ukj}
\end{figure}

\begin{Decomposition}
Consider the implementation of \( k \) unitary operations \( (\mathcal{U}_{k})^1, (\mathcal{U}_{k})^{2^1}, \dots, (\mathcal{U}_{k})^{2^{k-1}} \). These circuits are equivalent to those for \( \mathcal{U}_{k}, \mathcal{U}_{k-1}, \dots, \mathcal{U}_{1} = X \). Specifically, 
\begin{equation}
(\mathcal{U}_k)^{2^l} = \mathcal{U}_{k-l} \otimes I^{\otimes l}
\end{equation}   for \( l = 0, \dots, k-1 \), as illustrated in Fig. (\ref{fig:k-l}). Additionally, the circuit decomposition of \(\mathcal{U}_l\) is shown in Fig. (\ref{Ul}).  
\end{Decomposition}

\textbf{Analysis.} We have $(\mathcal{U}_k)^{2^l}=\sum_{m=0}^{2^k-1}|m -2^l \bmod 2^k\rangle\langle m|$. 
The binary form of $2^l$ is the following 
\begin{equation}
\vec{l}=(\underbrace{0, \cdots, 0}_{k-l-1}, 1, \underbrace{0, \cdots, 0}_{l}).
\end{equation} 
We express the binary form of $m$ as $\vec{m}=(m_0\cdots,m_{k-1})$. 
In the binary form, 
\begin{equation}
\begin{aligned}
(\mathcal{U}_k)^{2^l} &= \sum_{m_0,\cdots,m_{k-1}=0,1} |(\vec{m}-\vec{l}) \bmod 2^k\rangle\langle \vec{m}| \\
&=  \sum_{m_0,\cdots,m_{k-1}=0,1}  
|[(m_0\cdots m_{k-l-1})-(\underbrace{0, \cdots, 0}_{k-l-1}, 1)]\bmod 2^{k-l}\rangle\langle m_0\cdots m_{k-l-1}| \otimes  ( |m_{k-l}\cdots m_{k-1}\rangle\langle m_{k-l}\cdots m_{k-1}|)\\
&=\sum_{m=0}^{2^{k-l}-1}|m-1\bmod 2^{k-l}\rangle \langle m| \otimes I^{\otimes l}\\
&= \mathcal{U}_{k-l}\otimes I^{\otimes l}.
\end{aligned}
\end{equation}

 \begin{figure}[h]
    \centering
   \[
    \Qcircuit @C=1em @R=1em {
         & \lstick{q_1}   & \qw & \qw & \multigate{5}{(\mathcal{U}_k)^{2^l}} & \qw & \qw  &&&  &   & \qw& \qw &\multigate{2}{\mathcal{U}_{k-l} } & \qw& \qw\\
         & \vdots    &   &   &  &   &     &&&  &   &  &   & &  &  \\
         & \lstick{q_{k-l} }   & \qw & \qw & \ghost{(\mathcal{U}_k)^{2^l}} & \qw & \qw  &&=  & & &\qw & \qw &\ghost{\mathcal{U}_{k-l}}& \qw& \qw\\
         & \lstick{q_{k-l+1} }   & \qw & \qw & \ghost{(\mathcal{U}_k)^{2^l}} & \qw & \qw  &&   & & &\qw & \qw &\qw& \qw& \qw\\
         & \vdots & &  &   &  &    &&& &   &  &  & \vdots& &  \\
         & \lstick{q_{k} }   & \qw & \qw & \ghost{(\mathcal{U}_k)^{2^l}} & \qw & \qw   &&&  &  & \qw& \qw& \qw& \qw& \qw 
    }
\]
    \caption{Quantum Circuit for \((\mathcal{U}_k)^{2^l}\). Here, \(l\) ranges from \(0\) to \(k-1\). In the circuit on the right, the operation applied to the final \(l\) qubit is the identity operation \(I\). As \(l\) increases, the circuit simplifies.}
    \label{fig:k-l}
\end{figure}
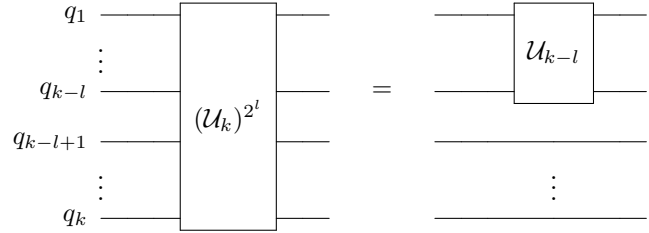

Now we consider the decomposition of $\mathcal{U}_l$, where $l=1,\cdots,k$. 
When $l=1$, $\mathcal{U}_1=\sum_{m=0}^1 |m-1 \bmod 2\rangle\langle m|=|0\rangle\langle 1|+|1\rangle\langle 0|=X$.  

For general \( l \), \(\mathcal{U}_l = \sum_{m=0}^{2^l-1} |m-1 \bmod 2^l\rangle\langle m|\) represents a global shift operation. The classical counterpart of this operation, the basic increment by 1 (`+1'), serves as its foundation. 
Its quantum counterpart is the fundamental shift operation of the Weyl-Heisenberg group. 
The circuit implementation of $\mathcal{U}_l$ is also used in pure state (rank-1 density matrix) QST \cite{wang2022pure}, as depicted in Fig. (\ref{Ul}). \qed
\begin{figure}[!htb]
 \[
 \quad\quad \Qcircuit @C=0.8em @R=0.8em {
\lstick{q_1}  &\multigate{4}{\quad \mathcal{U}_l\quad} &\qw&\\
\lstick{q_2}  &\ghost{\quad \mathcal{U}_l\quad} &\qw &\\
\lstick{\vdots}&  & & \\
\lstick{q_{l-1}}&\ghost{\quad \mathcal{U}_l\quad} &\qw &\\
\lstick{q_{l}}&\ghost{\quad \mathcal{U}_l\quad} &\qw &\\
}
\quad\quad
\Qcircuit @C=1em @R=1em{
& \quad\quad \\
& \quad\quad\\
& \quad\quad\\
& =\quad\quad\\
& \quad\quad\\
}
\quad
\Qcircuit @C=0.8em @R=0.8em {
&\qw  &\multigate{3}{ \mathcal{U}_{l-1}} &\qw&\qw\\
&\qw  &\ghost{\mathcal{U}_{l-1}} &\qw &\qw \\
&  & & \\
&\qw&\ghost{ \mathcal{U}_{l-1}} &\qw &\qw \\
&\gate{X}&\ctrl{-1} &\qw &\qw \\
}\]
\[
\quad
\Qcircuit @C=1em @R=1em{
& \quad\quad \\
& \quad\quad\\
& \quad\quad\\
& =\quad\quad\\
& \quad\quad\\
}
\Qcircuit @C=0.8em @R=0.8em{
    &\qw      &\qw       &\qw   &   \qw &   \targ      &\qw     \\
   &\qw         &\qw       &\qw & \targ &    \ctrl{-1}  &\qw        \\
 &         &          & \cdots  &        &  \\
&\qw      &\targ       &\qw &  \ctrl{-2} &   \ctrl{-2}  &\qw        \\
    &\gate{X} &\ctrl{-1} & \qw &  \ctrl{-1} &   \ctrl{-1}&\qw &
}
\quad
\Qcircuit @C=1em @R=1em{
& \quad\quad \\
& \quad\quad\\
& \quad\quad \\
& =\quad\quad\\
& \quad\quad\\
}
\quad
\Qcircuit @C=0.8em @R=0.8em {
    &   \targ      & \qw      &  \qw  & \qw      & \qw    & \qw \\
   &    \ctrlo{-1} & \targ    &  \qw  & \qw         & \qw  & \qw \\
  &              &          &   \cdots  &        &  \\
 &   \ctrlo{-2}  &  \ctrlo{-2}& \qw     & \targ     & \qw   & \qw    \\
   &   \ctrlo{-1}  &  \ctrlo{-1} & \qw    &  \ctrlo{-1} &\gate{X}   & \qw
}
\]
\caption{Quantum circuit for implementing $\mathcal{U}_{l}$. 
At the top, the relation between circuits of $\mathcal{U}_l$ and $\mathcal{U}_{l-1}$ is given. 
Two implementation methods are introduced. In the circuit, a solid point indicates that the control qubit is in the \(|1\rangle\) state, while a hollow point indicates that the control qubit is in the \(|0\rangle\) state.}
  \label{Ul}
\end{figure}
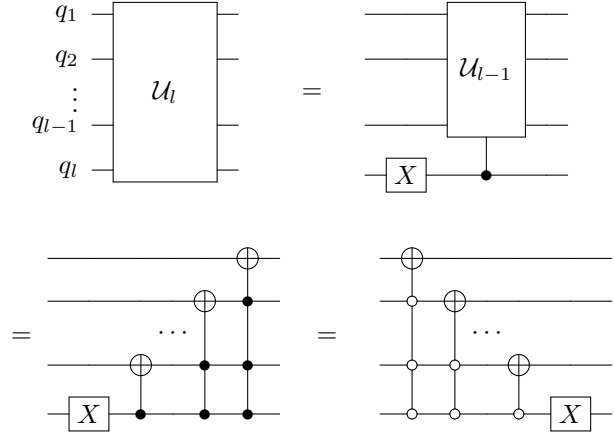

\begin{lemma}
    Each permutation operation required to perform PM on DDBs can be decomposed into at most \(O(n^4)\) elementary 1-qubit and 2-qubit gates.
\end{lemma}

Analysis: Denote \(\wedge_m(X)\) as the generalized Toffoli gate with \(m+1\) input bits, which maps \(\left|x_1, \ldots, x_m, y\right\rangle\) to \(\left|x_1, \ldots, x_m, \left(\prod_{k=1}^m x_k\right) \oplus y\right\rangle\). On input \((x_1, \ldots, x_k, y)\), the gate applies \(X\) to \(y\) if and only if \(\prod_{k=1}^m x_k = 1\). When \(m = 1\), \(\wedge_1(X)\) corresponds to the Controlled-NOT operation, expressed as \(\wedge_1(X) = |0\rangle\langle 0| \otimes I + |1\rangle\langle 1| \otimes X\).

Thus, the operation \(\mathcal{U}_l\) in Fig. (\ref{Ul}) is a combination of the following gates:
\[
X, \wedge_1(X), \dots, \wedge_{l-1}(X).
\]

According to Corollary 7.6 in \cite{barenco1995elementary}, the \(l\)-qubit gate \(\wedge_{l-1}(X)\) can be decomposed into \(\Theta(l^2)\) elementary 1-qubit and 2-qubit gates. Consequently, \(\mathcal{U}_l\) can be decomposed into \(O(l^3)\) elementary gates. For the controlled operation \( |0\rangle\langle 0| \otimes I + |1\rangle\langle 1| \otimes \mathcal{U}_l \), the cost in elementary gates also scales as \(O(l^3)\), as it involves a combination of \(\wedge_1(X), \dots, \wedge_{l-1}(X), \wedge_l(X)\).

Now, consider a worst-case scenario where the permutation operation \(T_{j,k}\) in Eq. (\ref{Pt}) incurs the maximum cost in terms of elementary gates. This occurs when \(k = n\). For each \(j = 1, \dots, 2^{n-1} - 1\), \(T_{j,k}\) can be implemented with at most \(n\) controlled operations using Decomposition \ref{decom1}. Therefore, the upper bound for decomposing all permutation operations involved in DDBs is \(O(n^4)\). \qed

Using \(l-1\) ancilla qubits, the \(\wedge_{l-1}(X)\) gate can be decomposed into \(O(l)\) elementary gates, as shown in \cite{Nielsen2002}. Consequently, each permutation operation in the DDB circuit can be efficiently implemented using at most \(O(n^3)\) elementary gates with the aid of ancilla qubits.  
It is noteworthy that, with the following decomposition and strategy, the cost of gates could be further slightly reduced. 
\begin{Decomposition}\label{decom3}
    The operations \((\mathcal{U}_l)^j\) and \((\mathcal{U}_l)^{2^l-j}\) can be implemented using the same number of gates, where \(j = 1, \dots, 2^l - 1\).
\end{Decomposition}

Analysis: We have \(\mathcal{U}_l = \sum_{m=0}^{2^l-1} |m-1 \bmod 2^l\rangle \langle m|\). Thus, \((\mathcal{U}_l)^j \cdot (\mathcal{U}_l)^{2^l-j} = I\). Therefore, if we perform the conjugate transpose circuit of \((\mathcal{U}_l)^j\), we obtain the circuit for \((\mathcal{U}_l)^{2^l-j}\). \qed

\textbf{Strategy:} By combining the analysis from Decompositions \ref{decom1} and \ref{decom3}, the circuit decomposition of \((\mathcal{U}_l)^j\) can be further simplified compared to the binary expression. For instance, when \( j = 2^l - 1 \), we should integrate the circuits for \(\mathcal{U}_l\), \((\mathcal{U}_l)^{2^1}\), \(\dots\), and \((\mathcal{U}_l)^{2^{l-1}}\) as described in Decomposition \ref{decom1}. However, according to Decomposition \ref{decom3}, it suffices to implement a single circuit for \((\mathcal{U}_l)^{\dag}\). As a result, the circuit components for \((\mathcal{U}_l)^{2^1}\), \(\dots\), and \((\mathcal{U}_l)^{2^{l-1}}\) can be omitted, leading to a more efficient implementation.

In general, we can define a finite set of integers:
\[
S = \{\pm 1, \pm 2^1, \dots, \pm 2^{l-1}\}.
\]
For any \( j \in \{1, \dots, 2^l - 1\} \), we can identify the minimal elements in the set \( S \) such that their sum equals \( j \). We then decompose \(\mathcal{U}_l^j\) according to these minimal elements in \( S \), rather than just using the binary form corresponding to \(S' = \{1, 2^1, \dots, 2^{l-1}\}\).

 \subsection{Three circuits for arbitrary DM element}

 Now we consider the three circuits of PMs onto DDBs to directly reconstruct arbitrary unknown DM element $\rho_{jk}$, $0\le j <k\le d-1$. 
 The basis for determining the diagonal element is the computational basis measurement,  Pauli measurement $Z_1Z_2\cdots Z_n$.
 The other two circuits are constructed in the following way. 

\begin{itemize}
    \item Write the binary representation of $j$ and $k$, which are $j_1j_2\cdots j_n$ and $k_1k_2\cdots k_n$ respectively. 
    \item Find the first different qubit of $|j_1j_2\cdots j_n\rangle$ and $|k_1k_2\cdots k_n\rangle$. 
    We may as well denote it as $q_s$. Namely, 
    \begin{equation}
    \left\{
    \begin{aligned}
        j&=j_1,\cdots,j_{s-1},j_s=0,j_{s+1},\cdots,j_n\\
        k&=k_1,\cdots,k_{s-1},k_s=1,k_{s+1},\cdots,k_n
    \end{aligned}
    \right.
\end{equation}
      Denote the difference between the binary numbers $j_{s+1},\cdots,j_n$ and $k_{s+1},\cdots,k_n$ as 
      \begin{equation}
          l=\sum_{m=s+1}^n (k_m-j_m)\times 2^{n-m}.  
      \end{equation}
    \item The permutational operation for $\rho_{jk}$ is defined by 
    \begin{equation}
        T_{j,k}=I^{\otimes s-1} \otimes [|0\rangle\langle 0|\otimes I+|1\rangle\langle 1|\otimes (\mathcal{U}_{n-s})^l]. 
    \end{equation}
\end{itemize}

The circuits for the PMs onto the nontrivial DDBs are depicted in Fig. (\ref{fig:rhojk}). When we go through all \(\rho_{jk}\), the required circuit types are \(O(2^n)\) instead of \(O(4^n)\). We can verify the function of the conjugate transpose of the circuits. 

After applying the operation \(H\) at qubit $q_s$, the state \(|j\rangle = |j_1 \cdots j_n\rangle\) evolves to 
\begin{equation}
    |j_1\cdots j_{s-1}\rangle \frac{|0\rangle + |1\rangle}{\sqrt{2}} |j_{s+1}\cdots j_n\rangle.
\end{equation}
Following the conditional permutation operation, the final state becomes 
\begin{equation}
    \frac{|j_1\cdots j_n\rangle + |k_1\cdots k_n\rangle}{\sqrt{2}} = |\phi^{+}_{jk}\rangle. 
\end{equation}
Thus, in the left circuit of Fig. (\ref{fig:rhojk}), if the measurement result is \(j_1,\cdots,j_n\), it corresponds to the projected state \(|\phi^{+}_{jk}\rangle\). Similarly, if the result is \(k_1,\cdots,k_n\), it corresponds to \(|\phi^{-}_{jk}\rangle\). 
In the right circuit, if the measurement result is \(j_1,\cdots,j_n\) (or \(k_1,\cdots,k_n\)), it corresponds to the projected state \(|\psi^{+}_{jk}\rangle\) (or \(|\psi^{-}_{jk}\rangle\)). 

    \begin{figure}[!htb]
  \[
 \Qcircuit @C=0.8em @R=0.8em {
 \lstick{q_1}  &\qw 	 &\qw &\qw  & \qw  &\meter  &&&&&&  &\qw 	 &\qw &\qw  & \qw  &\meter \\  
 \lstick{\vdots}     &&\vdots&&&& & &&&& &&\vdots\\  
\lstick{q_s}  &\qw 	 &\ctrl{1} &\qw  & \qw  &\measureD{\sigma_x}  &&&&&& &\qw 	 &\ctrl{1} &\qw  & \qw  &\measureD{\sigma_y} \\
\lstick{q_{s+1}}     & \qw & \multigate{3}{\mathcal{U}_{n-s}^{ l}} & \qw  & \qw  &\meter  &&&&&&  & \qw & \multigate{3}{\mathcal{U}_{n-s}^{ l}} & \qw  & \qw  &\meter\\
\lstick{\vdots}    &&&&&&   \\
\lstick{q_{n-1}}     & \qw &\ghost{\mathcal{U}_{n-s}^{ l}}& \qw  & \qw  &\meter &&&&&& & \qw &\ghost{\mathcal{U}_{n-s}^{ l}}& \qw  & \qw  &\meter\\
\lstick{q_n}     & \qw& \ghost{\mathcal{U}_{n-s}^{ l}} & \qw & \qw  &\meter  &&&&&& & \qw& \ghost{\mathcal{U}_{n-s}^{ l}} & \qw & \qw  &\meter}
 \]
  \caption{Circuits for PMs to determine \(\rho_{jk}\). At qubit \(q_s\), the measurements are Pauli \(X\) and \(Y\), denoted by \(\sigma_x\) and \(\sigma_y\), respectively. The measurement on the other qubit is a Pauli \(Z\) measurement. The circuit decomposition of \(T_{j,k}\) is discussed in Appendix \ref{sec:circuit}. In the left circuit, the measurement results \(j_1, \cdots, j_n\) and \(k_1, \cdots, k_n\) correspond to the projected state \(|\phi^{\pm}_{jk}\rangle\). In the right circuit, these results correspond to the projected state \(|\psi^{\pm}_{jk}\rangle\). 
Notably, while the projected states \((|\phi^+_{jk}\rangle, |\psi^+_{jk}\rangle)\) can be directly used to reconstruct \(\rho_{jk}\) as described in Eq. (3) of the main text, any pair of projected states—\((|\phi^+_{jk}\rangle, |\psi^-_{jk}\rangle)\), \((|\phi^-_{jk}\rangle, |\psi^+_{jk}\rangle)\), or \((|\phi^-_{jk}\rangle, |\psi^-_{jk}\rangle)\)—is also sufficient for this reconstruction.}
  \label{fig:rhojk}
\end{figure}
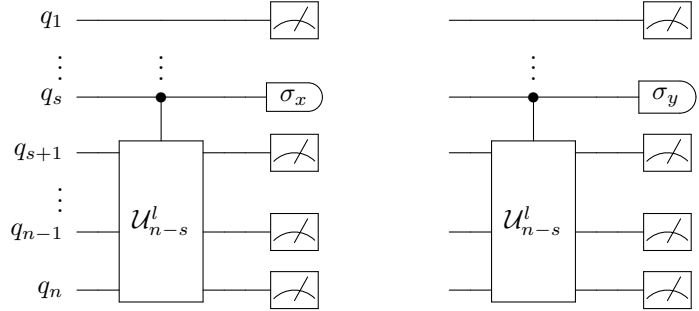

\section{Cloud experiments}
\label{experiment}
To test the performance of our strategy, real experiments were carried out on two quantum computers, superconducting qubits on IBM Quantum Lab and nuclear spins on Spinq.

For quantum chip on ibmq-manila, it is with a one-dimensional structure and 32 quantum volume, shown in SFig.(\ref{ibmq}). Table. (\ref{ibmq-manila parameters}) has its detailed parameters.
Only qubits 0 and 1 are used, with frequencies at 4.963 Ghz and 4.838 Ghz, with anharmonicity -0.34335 Ghz and -0.34621 Ghz, respectively. The error of CNOT gate from 0 on 1 is 6.437e-3 which costs 277.333ns, while the error of CNOT gate from 1 on 0 is 6.437e-3 which costs 312.889ns.

\begin{figure*}[!ht]
  \begin{center}
    \includegraphics[width=0.4\textwidth]{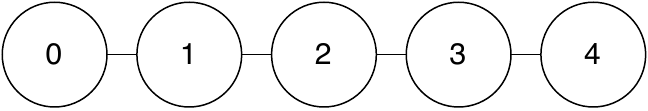}
    \caption{Structure for quantum chip on ibmq-manila: five superconducting transmon qubits. Only qubits $0$ and $1$, which are at $4.963$GHz and $4.838$GHz are employed, and its CNOT error is $6.437e$-$3$.}
    \label{ibmq}
    \end{center}
  \end{figure*}

\begin{table*}[!ht]
  \resizebox{\textwidth}{!}{%
  \begin{tabular}{|c|c|c|c|c|c|c|c|c|c|c|c|c|c|}
    \hline
Qubit & T1  &T2  &Frequency &Anharmonicity  &Readout  &Prob meas 0  &Prob meas 1  &Readout length  &ID   &  $\sqrt{x}$ (sx)   &Single-qubit  &CNOT   &Gate time  \\ 
 &  (us) & (us) & (GHz) & (GHz) & assignment error  & prep 1  & prep 0  & (ns) & error  &  error  & Pauli-X error  & error  & (ns) \\ 
\hline 
\multirow{2}{*}{Q0} &\multirow{2}{*}{130.3} &\multirow{2}{*}{91.35} &\multirow{2}{*}{4.963} &\multirow{2}{*}{-0.34335} &\multirow{2}{*}{0.0415} &\multirow{2}{*}{0.0644} &\multirow{2}{*}{0.0186} &\multirow{2}{*}{5351.111} &\multirow{2}{*}{0.0003787} &\multirow{2}{*}{0.0003787} &\multirow{2}{*}{0.0003787} &\multirow{2}{*}{$0_1:6.437e-3$} &\multirow{2}{*}{$0_1:277.333$}   \\ 
& & & & & & & & & & & & &  \\ \hline
\multirow{2}{*}{Q1} &\multirow{2}{*}{171.29} &\multirow{2}{*}{96.56} &\multirow{2}{*}{4.838} &\multirow{2}{*}{-0.34621} &\multirow{2}{*}{0.0206} &\multirow{2}{*}{0.0286} &\multirow{2}{*}{0.0126} &\multirow{2}{*}{5351.111} &\multirow{2}{*}{0.0001711} &\multirow{2}{*}{0.0001711} &\multirow{2}{*}{0.0001711} &$1_2:8.427e-3$ &$1_2:469.333$  \\ 
& & & & & & & & & & & &$1_0:6.437e-3 $&$1_0:312.889$  \\ \hline
\multirow{2}{*}{Q2} &\multirow{2}{*}{160.5} &\multirow{2}{*}{25.48} &\multirow{2}{*}{5.037} &\multirow{2}{*}{-0.34366} &\multirow{2}{*}{0.0185} &\multirow{2}{*}{0.0278} &\multirow{2}{*}{0.0092} &\multirow{2}{*}{5351.111} &\multirow{2}{*}{0.0002364} &\multirow{2}{*}{0.0002364} &\multirow{2}{*}{0.0002364} &$2_3:6.694e-3$ &$2_3:355.556$  \\  
 & & & & & & & & && & &$2_1:8.427e-3 $&$2_1:504.889$  \\  \hline
 \multirow{2}{*}{Q3} &\multirow{2}{*}{138.87} &\multirow{2}{*}{58.76} &\multirow{2}{*}{4.951} &\multirow{2}{*}{-0.34355} &\multirow{2}{*}{0.0213} &\multirow{2}{*}{0.0318} &\multirow{2}{*}{0.0108} &\multirow{2}{*}{5351.111} &\multirow{2}{*}{0.0002015} &\multirow{2}{*}{0.0002015} &\multirow{2}{*}{0.0002015} &$3_4:1.100e-2$ &$3_4:334.222$   \\  
 & & & & & & & & & & & &$3_2:6.694e-3 $&$3_2:391.111$   \\  \hline
 \multirow{2}{*}{Q4} &\multirow{2}{*}{143.02} &\multirow{2}{*}{46.67} &\multirow{2}{*}{5.066} &\multirow{2}{*}{-0.34211} &\multirow{2}{*}{0.0208} &\multirow{2}{*}{0.0308} &\multirow{2}{*}{0.0108} &\multirow{2}{*}{5351.111} &\multirow{2}{*}{0.0002743} &\multirow{2}{*}{0.0002743} &\multirow{2}{*}{0.0002743} &\multirow{2}{*}{$4_3:1.100e-2$} &\multirow{2}{*}{$4_3:298.667$}\\
 & & & & & & & & & & & & &  \\\hline
  \end{tabular}
  }
  \caption{ibmq-manila parameters}
  \label{ibmq-manila parameters}
\end{table*}

For entire experiments, states such as 
\begin{eqnarray}
  \ket{\psi_{1i}}&=&\alpha_i|0\rangle^{\otimes 2}+\beta_i|1\rangle^{\otimes 2} \label{S21} \\
  \ket{\psi_{2i}}&=&(\alpha_i|0\rangle+\beta_i|1\rangle)^{\otimes 2}
  \label{S22}
\end{eqnarray}
are tested, with i=1,...,21 and $\alpha_i=\cos(\theta_i/2)$, $\beta_i=\sin(\theta_i/2)$, $\theta_i= (i-1) \pi/20$. 
The prepared circuits are depicted in SFig.(\ref{2qubit_circuit}).
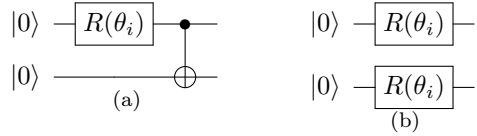
\begin{figure}[htbp]
  \centering
  \subfigure[]{
  \begin{minipage}[t]{0.2\linewidth}
  \centering
  \[
    \Qcircuit @C=0.8em @R=0.8em {
      \lstick{\ket{0}}  &\gate{R(\theta_i)} & \ctrl{1} & \qw\\
      \lstick{\ket{0}}  &\qw &\targ &\qw \\
  }
   \]
  %\caption{fig1}
  \end{minipage}%
  }%
  \subfigure[]{
  \begin{minipage}[t]{0.2\linewidth}
  \centering
   \[
    \Qcircuit @C=0.8em @R=0.8em {
      \lstick{\ket{0}}   & \gate{R(\theta_i)} & \qw \\
      \lstick{\ket{0}}  & \gate{R(\theta_i)} &   \qw}
  \]
  %\caption{fig2}
  \end{minipage}%
  }%
  \centering
  \caption{ Quantum circuits to prepare the test states. (a) is for $\ket{\psi_{1i}}$ and (b) is for  $\ket{\psi_{2i}}$, where $R(\theta_i)=\begin{pmatrix}
    \cos(\theta_i/2)) & -\sin(\theta_i/2))\\
    \sin(\theta_i/2)) & \cos(\theta_i/2))
  \end{pmatrix}$
  } \label{2qubit_circuit}
\end{figure}
With the construction method, $7$ measurement circuits are generated, which is shown in SFig.(\ref{ibm2qubit}).

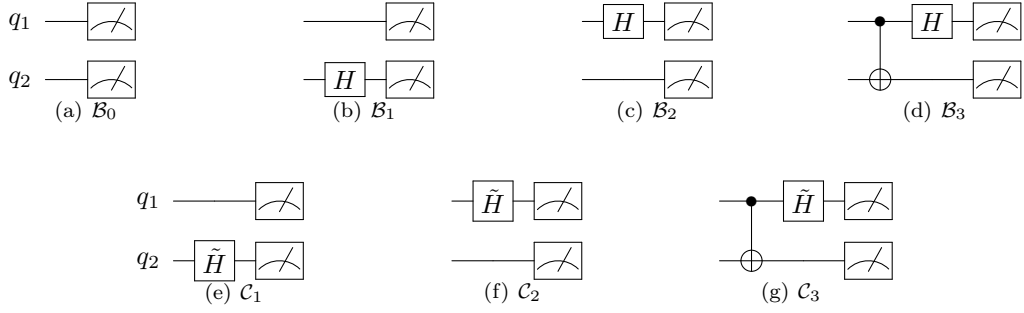
\begin{figure}[htbp]
  \centering
  \subfigure[$\mathcal{B}_0$]{
  \begin{minipage}[t]{0.2\linewidth}
  \centering
  \[
    \Qcircuit @C=0.8em @R=0.8em {
  \lstick{q_1}                      &\qw     &\meter\\
  \lstick{q_{2}}                    &\qw &\meter\\
  }
   \]
  %\caption{fig1}
  \end{minipage}%
  }%
  \subfigure[$\mathcal{B}_1$]{
  \begin{minipage}[t]{0.2\linewidth}
  \centering
   \[
    \Qcircuit @C=0.8em @R=0.8em {
       & \qw &  \meter \\
     & \gate{H} &  \meter }
  \]
  %\caption{fig2}
  \end{minipage}%
  }%
  \subfigure[$\mathcal{B}_2$]{
  \begin{minipage}[t]{0.2\linewidth}
  \centering
  \[
  \Qcircuit @C=0.8em @R=0.8em {
     & \gate{H} &  \meter  \\
    & \qw   &  \meter }
  \]
  %\caption{fig2}
  \end{minipage}
  }%
  \subfigure[$\mathcal{B}_3$]{
  \begin{minipage}[t]{0.2\linewidth}
  \centering
  \[
   \Qcircuit @C=0.8em @R=0.8em {
    &\ctrl{1}  & \gate{H} &  \meter  \\
    &\targ   &\qw &  \meter }
   \]
  %\caption{fig2}
  \end{minipage}
  }%
  
  \centering
  \subfigure[$\mathcal{C}_1$]{
  \begin{minipage}[t]{0.2\linewidth}
  \centering
   \[
    \Qcircuit @C=0.8em @R=0.8em {
      \lstick{q_1}                      &\qw     &\meter\\
      \lstick{q_{2}}                    &\gate{\tilde{H}}&\meter\\
      }
  \]
  %\caption{fig1}
  \end{minipage}%
  }%
  \centering
  \subfigure[$\mathcal{C}_2$]{
  \begin{minipage}[t]{0.2\linewidth}
  \centering
  \[
  \Qcircuit @C=0.8em @R=0.8em {
     & \gate{\tilde{H}} &  \meter  \\
    & \qw   &  \meter }
  \]
  %\caption{fig1}
  \end{minipage}%
  }%
  \centering
  \subfigure[$\mathcal{C}_3$]{
  \begin{minipage}[t]{0.2\linewidth}
  \centering
  \[
   \Qcircuit @C=0.8em @R=0.8em {
   &\quad &\ctrl{1}  & \gate{\tilde{H}} &  \meter  \\
   &\quad &\targ   &\qw &  \meter }
   \]
  %\caption{fig1}
  \end{minipage}%
  }%
  \centering
  \caption{The circuits for measurement. 
  The measurement on each qubit is the projective measurement onto basis $\{|0\rangle,|1\rangle\}$. 
  In front of the circuits, the input quantum states to be determined is ignored. 
  $H$ is the hadamard gate while $\tilde{H}$ is $U(\frac{\pi}{2},0,\frac{\pi}{2})$ in ibmq's setting.
%The measurement at each qubit is the canonical one, Pauli measurement $\sigma_z$.
} 
\label{ibm2qubit}
\end{figure}

For the entire experiment, two sets were conducted. The first is through our strategy with data analysis of direct calculation and SDP. The second is the standard tomography protocol. The main results such as Frobenius's distance and fidelities are calculated in the manuscript. Here we list the density matrix for each experiment-prepared state as supplementary. 
SFig.(\ref{dm1}) and SFig.(\ref{dm2}) are from Eq.(\ref{S21}) and Eq.(\ref{S22}), respectively. Although 21 experiments were conducted, only 11 density matrices are listed from $0$ to $\pi$.
SFig.(\ref{dm1}) (SFig.(\ref{dm2})) is divided into two lines, the first line is the real parts of the density matrix and the second line is the image parts.
Meanwhile, transparency parts are theoretical values, and solid parts are from experiments. 
 \begin{figure*}[!h]
  \begin{center}
    \includegraphics[width=1.\textwidth]{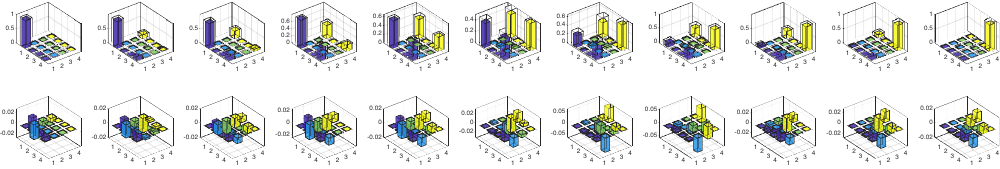}
    \caption{Density matrix for $\alpha_i|0\rangle^{\otimes 2}+\beta_i|1\rangle^{\otimes 2}$, where $\alpha_i=\cos((i-1)*\pi/10)$ is presented.The first line is the real part while the second line is the image part. The transparency part is the theoretical comparison. }
    \label{dm1}
    \end{center}
  \end{figure*}

  \begin{figure*}[!ht]
    \begin{center}
      \includegraphics[width=1.\textwidth]{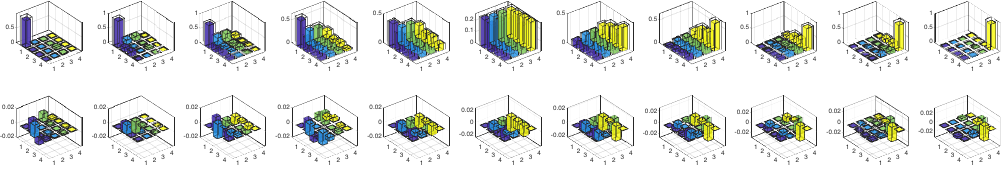}
      \caption{Density matrix for $(\alpha_i|0\rangle+\beta_i|1\rangle)^{\otimes 2}$, where $\alpha_i=\cos((i-1)\pi/10)$ is presented.The first line is the real part while the second line is the image part. The transparency part is the theoretical comparison. }
      \label{dm2}
      \end{center}
    \end{figure*}

As for the cloud quantum computer of SpinQ, it is a liquid NMR-based architecture, that uses crotonic acid as their qubit system. 
As it is shown in SFig.(\ref{molecule_spinq}). 4 carbon nuclei are denoted as 4 qubits, where related parameters are listed in the table.
With an external programmable radio-frequency pulse as a control field, almost 4-qubit quantum logic gates can be achieved. 
In the table, diagonal elements are frequencies, while off-diagonal elements are J-couplings, which are all measured at room temperature. 

\begin{figure*}[!ht]
  \begin{center}
    \includegraphics[width=0.6\columnwidth]{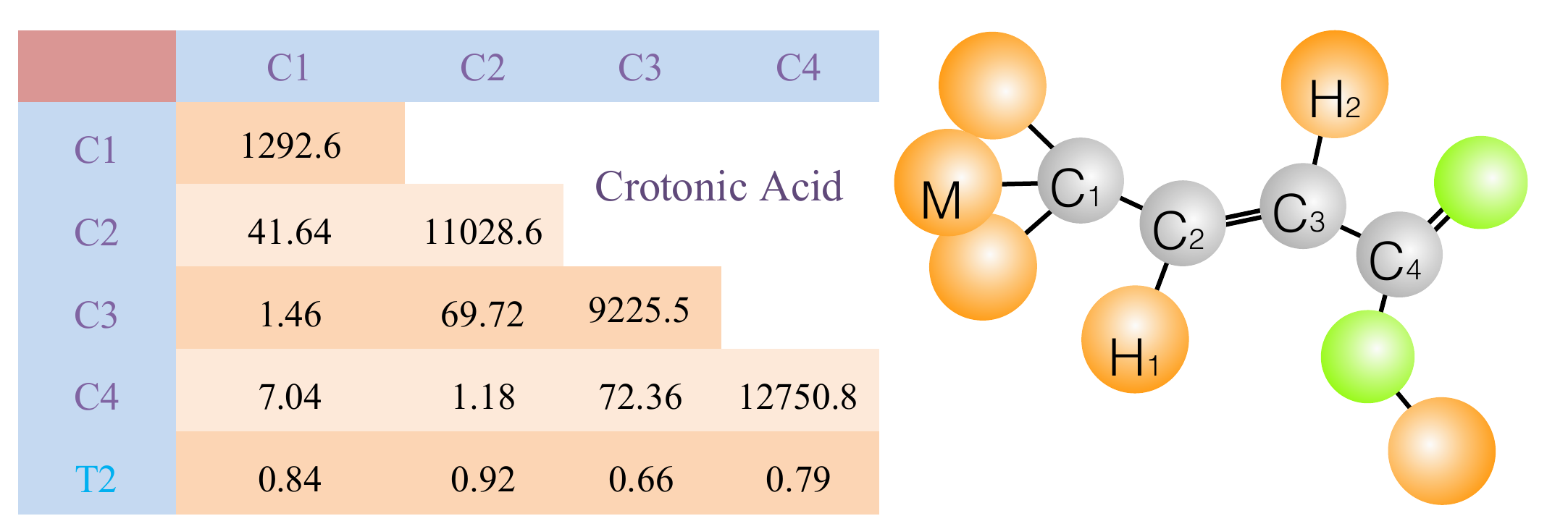}
    \caption{Structure for quantum chip on Spinq: four carbon nuclear spins. In the table, diagonal elements are frequencies, and off-diagonal elements are J-couplings.} 
    \label{molecule_spinq}
    \end{center}    
  \end{figure*}

In order to demonstrate our strategy, $30$ measurement circuits are generated, which is shown in SFig.(\ref{4-qubit}). $\mathcal{B}_0$ is ignored here as it is a trivial one as conventional tomography strategy. However, we do not realize the following circuits, since some entangled gates are out of the current device's capability. Decoherence time doesn't allow us more quantum gates. 
Thus, we have to simulate circuits in SFig.(\ref{4-qubit}) by decomposing each measurement basis into Pauli and summing them up in the end, indirectly completing the proposal.  

\begin{figure}[htbp]
  \subfigure[$\mathcal{B}_1$]{
  \begin{minipage}[t]{0.05\linewidth}
  \centering
   \[
    \Qcircuit @C=0.8em @R=0.8em {
      \lstick{q_1}             &\meter\\
      \lstick{q_2}             &\meter \\
      \lstick{q_3}             &\meter \\
      \lstick{q_4}              &\measureD{\sigma_x} \\ }
  \]
  \end{minipage}%
  }%
  \subfigure[$\mathcal{B}_2$]{
  \begin{minipage}[t]{0.08\linewidth}
  \centering
  \[
  \Qcircuit @C=0.8em @R=0.8em {
                    &\meter\\
                   &\meter \\
                   &\measureD{\sigma_x} \\
                    &\meter \\}
  \]
  \end{minipage}
  }%
  \subfigure[$\mathcal{B}_3$]{
    \begin{minipage}[t]{0.08\linewidth}
    \centering
    \[
      \Qcircuit @C=0.8em @R=0.8em {
    \lstick{}              &\qw        &\meter\\
    \lstick{}              &\qw       &\meter \\
      \lstick{}            &\ctrl{1}         &\measureD{\sigma_x} \\
        \lstick{}           &\targ         &\meter \\
    }
     \]
    \end{minipage}%
    }%
  \subfigure[$\mathcal{B}_4$]{
  \begin{minipage}[t]{0.08\linewidth}
  \centering
   \[
    \Qcircuit @C=0.8em @R=0.8em {
                    &\meter\\
              &\measureD{\sigma_x} \\
                     &\meter \\
                  &\meter \\ }
  \]
  \end{minipage}%
  }%
  \subfigure[$\mathcal{B}_5$]{
    \begin{minipage}[t]{0.12\linewidth}
    \centering
    \[
      \Qcircuit @C=0.8em @R=0.8em {
             &\qw       &\meter\\
           &\ctrl{1}      &\measureD{\sigma_x} \\
         &\multigate{1}{  \mathcal{U}^1_{2}}         &\meter \\
       &\ghost {  \mathcal{U}^1_{2}}             &\meter \\
    }
     \]
    \end{minipage}% 
    }%
    \subfigure[$\mathcal{B}_6$]{
      \begin{minipage}[t]{0.12\linewidth}
      \centering
       \[
        \Qcircuit @C=0.8em @R=0.8em {
                &\qw          &\meter\\
             &\ctrl{1}           &\measureD{\sigma_x} \\
            &\multigate{1}{  \mathcal{U}^2_{2}}          &\meter \\
          &\ghost {  \mathcal{U}^2_{2}}            &\meter \\ }
      \]
      \end{minipage}%
      }%  
  \subfigure[$\mathcal{B}_7$]{
  \begin{minipage}[t]{0.12\linewidth}
  \centering
  \[
  \Qcircuit @C=0.8em @R=0.8em {
          &\qw         &\meter\\
         &\ctrl{1}        &\measureD{\sigma_x} \\
          &\multigate{1}{  \mathcal{U}^3_{2}}               &\meter \\
          &\ghost {  \mathcal{U}^3_{2}}              &\meter \\}
  \]
  \end{minipage}
  }%    
  \subfigure[$\mathcal{B}_8$]{
  \begin{minipage}[t]{0.08\linewidth}
  \centering
  \[
   \Qcircuit @C=0.8em @R=0.8em {
                    &\measureD{\sigma_x}\\
           &\meter \\
                  &\meter \\
                    &\meter \\}
   \]
  \end{minipage}
  }%
  \subfigure[$\{\mathcal{B}_j;j=9,\cdots,15\}$ with $k=1,\cdots,7$]{
    \begin{minipage}[t]{0.15\linewidth}
    \centering
      \[
      \Qcircuit @C=0.8em @R=0.8em {
    \lstick{}             &\ctrl{1}         &\measureD{\sigma_x}\\
    \lstick{}           &\multigate{2}{ ~~~~ \mathcal{U}^k_{3}~~~}         &\meter \\
      \lstick{}         &\ghost {  ~~~~ \mathcal{U}^k_{3}~~~}         &\meter \\
        \lstick{}        &\ghost { ~~~~ \mathcal{U}^k_{3}~~~}             &\meter \\
    }
     \]
    \end{minipage}%
    }%
  \\
    \centering
  \subfigure[$\mathcal{C}_1$]{
  \begin{minipage}[t]{0.05\linewidth}
  \centering
   \[
    \Qcircuit @C=0.8em @R=0.8em {
      \lstick{q_1}     &\meter\\
      \lstick{q_2}    &\meter \\
      \lstick{q_3}      &\meter \\
      \lstick{q_4}    &\measureD{\sigma_y} \\ }
  \]
  \end{minipage}%
  }%
  \centering
  \subfigure[$\mathcal{C}_2$]{
  \begin{minipage}[t]{0.08\linewidth}
  \centering
  \[
  \Qcircuit @C=0.8em @R=0.8em {
                      &\meter\\
                &\meter \\
                   &\measureD{\sigma_y} \\
                      &\meter  }
  \]
  \end{minipage}%
  }%
  \subfigure[$\mathcal{C}_3$]{
    \begin{minipage}[t]{0.08\linewidth}
    \centering
     \[
      \Qcircuit @C=0.8em @R=0.8em {
              &\qw        &\meter\\
             &\qw         &\meter \\
           &\ctrl{1}         &\measureD{\sigma_y} \\
               &\targ         &\meter \\ }
    \]
    \end{minipage}%
    }%
  \centering
  \subfigure[$\mathcal{C}_4$]{
  \begin{minipage}[t]{0.08\linewidth}
  \centering
   \[
    \Qcircuit @C=0.8em @R=0.8em {
                    &\meter\\
                  &\measureD{\sigma_y} \\
                   &\meter \\
                    &\meter \\ }
  \]
  \end{minipage}%
  }%
  \subfigure[$\mathcal{C}_5$]{
    \begin{minipage}[t]{0.12\linewidth}
    \centering
     \[
      \Qcircuit @C=0.8em @R=0.8em {
           &\qw         &\meter\\
             &\ctrl{1}           &\measureD{\sigma_y} \\
             &\multigate{1}{  \mathcal{U}^1_{2}}        &\meter \\
             &\ghost {  \mathcal{U}^1_{2}}          &\meter \\ }
    \]
    \end{minipage}%
    }%  
  \subfigure[$\mathcal{C}_6$]{
  \begin{minipage}[t]{0.12\linewidth}
  \centering
  \[
  \Qcircuit @C=0.8em @R=0.8em {
          &\qw            &\meter\\
       &\ctrl{1}            &\measureD{\sigma_y} \\
        &\multigate{1}{  \mathcal{U}^2_{2}}                 &\meter \\
       &\ghost {  \mathcal{U}^2_{2}}            &\meter  }
  \]
  \end{minipage}%
  }%
  \subfigure[$\mathcal{C}_7$]{
  \begin{minipage}[t]{0.12\linewidth}
  \centering
   \[
    \Qcircuit @C=0.8em @R=0.8em {
           &\qw             &\meter\\
       &\ctrl{1}              &\measureD{\sigma_y} \\
        &\multigate{1}{  \mathcal{U}^3_{2}}               &\meter \\
        &\ghost {  \mathcal{U}^3_{2}}              &\meter \\ }
  \]
  \end{minipage}%
  }%
  \subfigure[$\mathcal{C}_8$]{
  \begin{minipage}[t]{0.08\linewidth}
  \centering
  \[
   \Qcircuit @C=0.8em @R=0.8em {
                   &\measureD{\sigma_y}\\
                  &\meter \\
                  &\meter \\
                  &\meter  }
   \]
  \end{minipage}%
  }%
  \subfigure[$\{\mathcal{C}_j;j=9,\cdots,15\}$ with $k=1,\cdots,7$]{
  \begin{minipage}[t]{0.15\linewidth}
  \centering
   \[
    \Qcircuit @C=0.8em @R=0.8em {
   \lstick{}             &\ctrl{1}           &\measureD{\sigma_y}\\
  \lstick{}             &\multigate{2}{~~~~ \mathcal{U}^k_{3}~~~}         &\meter \\
    \lstick{}             &\ghost {  ~~~~ \mathcal{U}^k_{3}~~~}         &\meter \\
      \lstick{}            &\ghost {  ~~~~ \mathcal{U}^k_{3}~~~}          &\meter \\ }
  \]
  \end{minipage}%
  }%  
\caption{Circuits for 4-qubit tomography.  
%There are 9 circuits with no entanglement, 2 circuits with 2-qubit entanglement, 6 circuits with 3-qubit entanglement, and 20 circuits with 4-qubit entanglement.  
%The operation $\mathcal{V}_{j}$ is an $j$-qubit operation, $\mathcal{V}_{j}=(\mathcal{U}_{j})^{\dag}=\sum_{t=0}^{2^j-1}|t-1\bmod 2^{j}\rangle\langle t|$, $\mathcal{V}_{j}^k=(\mathcal{V}_{j})^k$.
The symbols $\sigma_x,\sigma_y$ mean Pauli measurement $X$, $Y$. While the remaining measurement at each qubit is Pauli $Z$ measurement, projected to $\{|0\rangle,|1\rangle\}$. } 
  \label{4-qubit}
\end{figure}
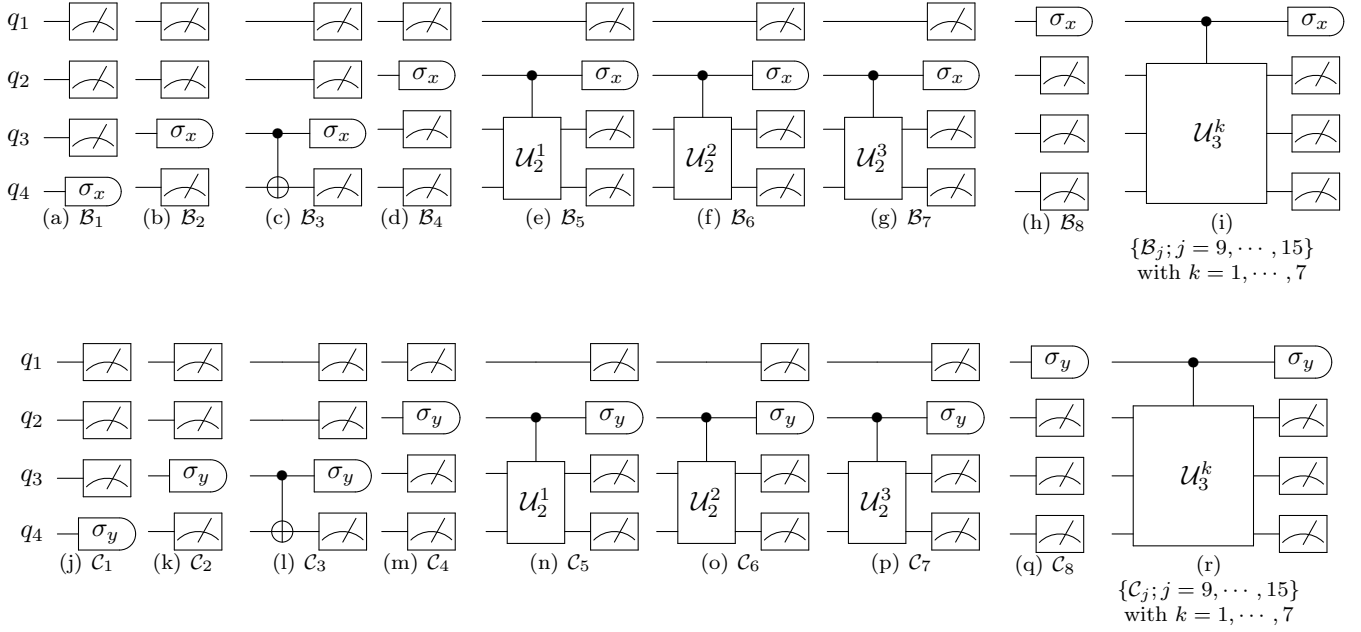

Therefore, For entire experiments, states such as 
\begin{eqnarray}
  \ket{\psi_{1i}}&=&\alpha_i|0\rangle^{\otimes 4}+\beta_i|1\rangle^{\otimes 4} \label{Spinq1} \\
  \ket{\psi_{2i}}&=&(\alpha_i|0\rangle+\beta_i|1\rangle)^{\otimes 4}
  \label{Spinq2}
\end{eqnarray}
are tested, with i=1,...,11 and $\alpha_i=\cos(\theta_i/2)$, $\beta_i=\sin(\theta_i/2)$, $\theta_i= (i-1) \pi/10$. 
The prepared circuits are depicted in SFig.(\ref{4qubit_circuit}).
\begin{figure}[!h]
  \centering
  \subfigure[]{
  \begin{minipage}[t]{0.2\linewidth}
  \centering
  \[
    \Qcircuit @C=0.8em @R=1.1em {
      \lstick{\ket{0}}  &\gate{R(\theta_i)} & \ctrl{1} & \ctrl{2} &\ctrl{3}&\qw\\
      \lstick{\ket{0}}  &\qw &\targ &\qw &\qw &\qw \\
      \lstick{\ket{0}}  &\qw &\qw &\targ &\qw &\qw\\
      \lstick{\ket{0}}  &\qw &\qw &\qw &\targ &\qw\\
  }
   \]
  %\caption{fig1}
  \end{minipage}%
  }%
  \subfigure[]{
  \begin{minipage}[t]{0.2\linewidth}
  \centering
   \[
    \Qcircuit @C=0.8em @R=0.4em {
      \lstick{\ket{0}}   & \gate{R(\theta_i)} & \qw \\
      \lstick{\ket{0}}  & \gate{R(\theta_i)} &   \qw\\
      \lstick{\ket{0}}   & \gate{R(\theta_i)} & \qw \\
      \lstick{\ket{0}}  & \gate{R(\theta_i)} &   \qw
      }
  \]
  %\caption{fig2}
  \end{minipage}%
  }%
  \centering
  \caption{ Quantum circuits to prepare the test states. (a) is for $\ket{\psi_{1i}}$ and (b) is for  $\ket{\psi_{2i}}$, where $R(\theta_i)=\begin{pmatrix}
    \cos(\theta_i/2)) & -\sin(\theta_i/2))\\
    \sin(\theta_i/2)) & \cos(\theta_i/2))
  \end{pmatrix}$
  } \label{4qubit_circuit}
\end{figure}
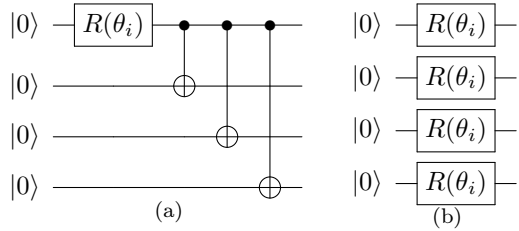
Similarly, results are presented with direct calculation and SDP. As a comparison, standard tomography protocol was also conducted. The main results of Frobenius's distance and fidelities are calculated in the manuscript. Here we only list the density matrix for each experiment-prepared state.

SFig.(\ref{spinq1}) and SFig.(\ref{spinq2}) are from Eq.(\ref{Spinq1}) and Eq.(\ref{Spinq2}), respectively. Although 11 experiments were conducted, only 6 density matrices are listed, where $\theta_i=0, \pi/5, 2\pi/2, 3\pi/5, 4\pi/5, \pi$.
SFig.(\ref{spinq1}) (SFig.(\ref{spinq2})) is divided into two lines, the first line is the real parts of the density matrix and the second line is the image parts.
Meanwhile, transparency parts are theoretical values, and solid parts are from experiments. As GHZ-like states were prepared through 3-CNOT gates, which cost around 100ms, the decoherence affects the states heavily. 

\begin{figure*}[!htbp]
  \begin{center}
    \includegraphics[width=1.\textwidth]{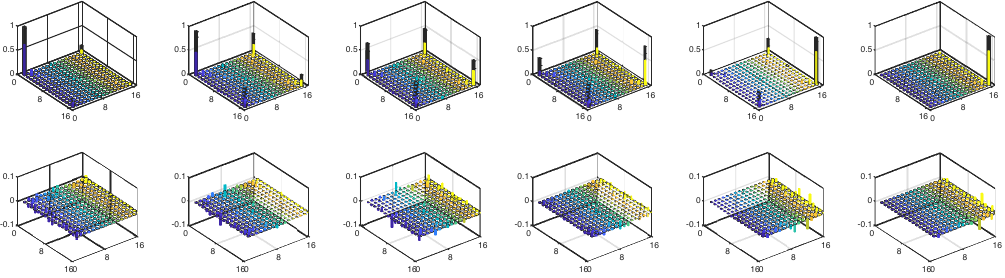}
    \caption{Density matrix for $\alpha_i|0\rangle^{\otimes 2}+\beta_i|1\rangle^{\otimes 2}$, where $\alpha_i=\cos((i-1)*\pi/10)$ is presented.The first line is the real part while the second line is the image part. The transparency part is the theoretical comparison. }
    \label{spinq1}
    \end{center}
  \end{figure*}
  \begin{figure*}[!htbp]
    \begin{center}
      \includegraphics[width=1.\textwidth]{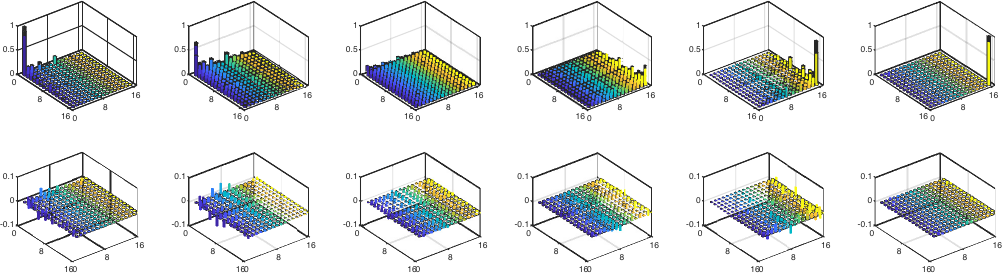}
      \caption{Density matrix for $(\alpha_i|0\rangle+\beta_i|1\rangle)^{\otimes 4}$, where $\alpha_i=\cos((i-1)*\pi/5)$ is presented.The first line is the real part while the second line is the image part. The transparency part is the theoretical comparison. }
      \label{spinq2}
      \end{center}
    \end{figure*}

\subsection*{Table of fidelity for entire experiments}
At the end of this section, we list the table for the entire experiments in Table.(\ref{table_fid}).
\begin{table*}[!ht]
  \resizebox{0.975\textwidth}{!}{%
  \begin{tabular}{|c|c|c|c|c|c|c|c|c|c|c|c|c|c|c|c|c|c|c|c|c|c|c|}
    \hline
 ${T}_1$& 1  &2  &3 &4  &5 &6  &7  &8  &9   & 10  &11 &12  &13 &14  &15 &16  &17  &18  &19   & 20  &21  \\ 
 \hline 
 Method$_d$& 0.963 &	0.975 &	0.983 &	0.980 &	0.976 &	0.978 &	0.970 &	0.966 &	0.962 &	0.954 &	0.956 &	0.947 &	0.945 &	0.945 &	0.943 &	0.933 &	0.931 &	0.922 &	0.925 &	0.915 &	0.907 \\ \hline
 Method$_s$& 0.970 &	0.980 &	0.987 &	0.986 &	0.987 &	0.981 &	0.980 &	0.981 &	0.982 &	0.984 &	0.983 &	0.984 &	0.986 &	0.983 &	0.981 &	0.979 &	0.964 &	0.951 &	0.939 &	0.928 &	0.918 \\ \hline
 Method$_t$& 0.964 &	0.953 &	0.949 &	0.962 &	0.946 &	0.940 &	0.956 &	0.945 &	0.933 &	0.936 &	0.874 &	0.904 &	0.830 &	0.853 &	0.888 &	0.892 &	0.931 &	0.928 &	0.925 &	0.921 &	0.926  \\ 
 \hline
 \hline
 ${T}_2$& 1  &2  &3 &4  &5 &6  &7  &8  &9   & 10  &11 &12  &13 &14  &15 &16  &17  &18  &19   & 20  &21  \\ 
 \hline
 Method$_d$& 0.909 & 0.956 & 0.957 &	0.955 &	0.959 &	0.957 &	0.955 &	0.955 &	0.954 &	0.947 &	0.949 &	0.954 &	0.941 &	0.940 &	0.951 &	0.941 &	0.934 &	0.928 &	0.920 &	0.918 &	0.906 \\ \hline
 Method$_s$& 0.929 &	0.978 &	0.984 &	0.989 &	0.982 &	0.988 &	0.984 &	0.983 &	0.989 &	0.990 &	0.982 &	0.980 &	0.976 &	0.970 &	0.969 &	0.961 &	0.950 &	0.941 &	0.930 &	0.927 &	0.914 \\ \hline
 Method$_t$& 0.923 &	0.963 &	0.960 &	0.950 &	0.959 &	0.906 &	0.935 &	0.963 &	0.961 &	0.959 &	0.977 &	0.924 &	0.953 &	0.942 &	0.953 &	0.952 &	0.939 &	0.935 &	0.941 &	0.931 &	0.920 \\ \hline 
 \hline
  \end{tabular}
  }
  \resizebox{0.5\textwidth}{!}{%
  \begin{tabular}{|c|c|c|c|c|c|c|c|c|c|c|c|}
    \hline
 ${F}_1$& 1  &2  &3 &4  &5 &6  &7  &8  &9   & 10 & 11   \\ 
 \hline
 Method$_d$& 0.715 &	0.700 &	0.630 &	0.602 &	0.586 &	0.567 &	0.608 &	0.598 &	0.656 &	0.759 &	0.775 \\ \hline
 Method$_s$& 0.710 &	0.716 &	0.622 &	0.578 &	0.528 &	0.515 &	0.544 &	0.543 &	0.639 &	0.777 &	0.796 \\ \hline
 Method$_t$& 0.664 &	0.641 &	0.562 &	0.540 &	0.518 &	0.516 &	0.527 &	0.504 &	0.557 &	0.704 &	0.700 \\ \hline 
 \hline
 ${F}_2$& 1  &2  &3 &4  &5 &6  &7  &8  &9   & 10 & 11   \\ 
 \hline 
 Method$_d$& 0.963 &	0.962 &	0.960 &	0.943 &	0.886 &	0.912 &	0.946 &	0.968 &	0.980 &	0.891 &	0.973 \\ \hline
 Method$_s$& 0.923 &	0.899 &	0.918 &	0.836 &	0.771 &	0.746 &	0.756 &	0.848 &	0.896 &	0.832 &	0.875 \\ \hline
 Method$_t$& 0.870 &	0.863 &	0.832 &	0.809 &	0.733 &	0.724 &	0.762 &	0.819 &	0.845 &	0.751 &	0.874 \\ \hline
  \end{tabular}
  }
  \caption{Fidelity for entire experiments, Method$_d$ and Method$_s$ are via our protocol with direct reconstruction and semi-definite programming; Method$_t$ is via traditional tomography protocol. $T_j$ and $F_j$ are for the two and four qubits experiment, respectively, where $j=1,2$ are for $\ket{\psi_1i}$ and $\ket{\psi_2i}$} and $i$ is for the  horizontal label. 
  \label{table_fid}
\end{table*}

\section{Error analysis}
A real situation is that measurement bases $\{\mathcal{B}_0, \mathcal{B}_k, \mathcal{C}_k, 1\le k\le d \}$ cannot be perfectly realized. In most cases, we could create very close ones, which cause slight differences when obtaining the measured probability. That is to say, possibly we project an unknown target $\rho$ onto an approximate basis state $\tilde{\ket{\phi_i}}$, instead of the exact one, $\ket{\phi_i}$. 
For certain ideal and realized basis operators, the $i$-th basis states (eigenstates) have such relation,  
\begin{eqnarray}
  \tilde{\ket{\phi_i}}=\ket{\phi_i}+\epsilon_i\ket{e_i},
\end{eqnarray}
where $\ket{\phi_i}$ is the exact $i$-th basis state, $\epsilon_i$ is a constant amplitude, and $\ket{e_i}$ is one random state of haar measure. 
Obviously, vast repeated measurements produce averaged effects,  
\begin{eqnarray}
  \int \ket{e_i} d e_i =0, \quad \int \ket{e_i}\bra{e_i} d e =I/d.
\end{eqnarray}
Therefore, probabilities measured on $\rho$ is with such disturbance,
\begin{eqnarray}
  \mbox{tr}(\rho\tilde{\ket{\phi_i}}\tilde{\bra{\phi_i}})=\frac{1}{1+\epsilon_i^2} \mbox{tr}(\rho\ket{\phi_i}\bra{\phi_i})+\frac{\epsilon_i^2}{1+\epsilon_i^2} \frac{1}{d}.
\end{eqnarray}

As for the procedure to reconstruct a certain density matrix $d$-dimension $\rho$, the trace distance is employed to show the performance of protocols under the above error assumption. Specifically, 
\begin{eqnarray}
  ||\rho-\sigma||_2=\mbox{tr}[(\rho-\sigma)\cdot (\rho-\sigma)^{\dagger}]=\sum_{i,j} \xi_{ij}\xi_{ij}^{\star},
\end{eqnarray}
where $\sigma$ is the reconstructed matrix and $\xi_{ij}=\rho_{ij}-\sigma_{ij}$ and $i,j \in[1, d]$.
Additionally, $\epsilon_i$ are assumed to be at the same level, i.e, $\epsilon$.

Therefore, for diagonal elements, $|\rho_{ii}|$ and its measured deviation,
\begin{eqnarray}
 \sum_{i} \xi_{ii}\xi_{ii}^{\star}&=&\sum_{i}|-\frac{\epsilon^2}{1+\epsilon^2}\rho_{ii}+\frac{\epsilon^2}{1+\epsilon^2} \frac{1}{d}|^2 \nonumber \\
 &\leq& 2 \sum_{i}|\frac{\epsilon^2}{1+\epsilon^2}\rho_{ii}|^2+ 2\sum_{i}|\frac{\epsilon^2}{1+\epsilon^2} \frac{1}{d}|^2 \nonumber \\
 &\sim& \mathcal{O}(\epsilon^4)
\end{eqnarray}
Here, the coefficients ahead are ignored as we assumed $\sum_i |\rho_{ii}|^2$ is bounded.
For off-diagonal elements, 
\begin{eqnarray}
  \rho_{jk} &=&  \mbox{tr}(\rho|k\rangle\langle j|)\nonumber \\
  &=&\mbox{tr}(\rho|\phi_{jk}^{+}\rangle\langle\phi_{jk}^{+}|)-i\mbox{tr}(\rho|\psi_{jk}^{+}\rangle\langle\psi_{jk}^{+}|)-\frac{1-i}{2}(\rho_{kk}+\rho_{jj}), \label{offdiagonal}
\end{eqnarray}
where notations in Eq.(\ref{offdiagonal}) are listed in main text.
As with multivariable derivative formula, $|\xi_{jk}|$ is bounded, where 
\begin{eqnarray}
   |\xi_{jk}|^2&\leq& \Delta^2 \mbox{tr}(\rho|\phi_{jk}^{+}\rangle\langle\phi_{jk}^{+}|)+ \Delta^2\mbox{tr}(\rho|\psi_{jk}^{+}\rangle\langle\psi_{jk}^{+}|)+|\xi_{kk}|^2+|\xi_{jj}|^2,
\end{eqnarray}
and 
\begin{eqnarray}
  \Delta^2 \mbox{tr}(\rho|\phi_{jk}^{+}\rangle\langle\phi_{jk}^{+}|)\leq 2\times (\frac{\epsilon^2}{1+\epsilon^2}\mbox{tr}(\rho|\phi_{jk}^{+}\rangle\langle\phi_{jk}^{+}|))^2+2\times (\frac{\epsilon^2}{1+\epsilon^2} \frac{1}{d})^2.
\end{eqnarray}
Accordingly, an approximate error is evaluated, 
\begin{eqnarray}
  \sum_{j\neq k}|\xi_{jk}|^2 \leq \mathcal{O}(\epsilon^4)\{\sum_{j\neq k} \mbox{tr}(\rho|\phi_{jk}^{+}\rangle\langle\phi_{jk}^{+}|)^2+\sum_{j\neq k} \mbox{tr}(\rho|\psi_{jk}^{+}\rangle\langle\psi_{jk}^{+}|)^2\}+\mathcal{O}(\epsilon^4)
\end{eqnarray}
For normalization condition,  ${\sum_{j\neq k} \mbox{tr}(\rho|\phi_{jk}^{+}\rangle\langle\phi_{jk}^{+}|}\sim \mathcal{O}(1)$. In summary, 
\begin{eqnarray}
  \sum_{j\neq k}|\xi_{jk}|^2 \sim \mathcal{O}(\epsilon^4).
\end{eqnarray}
Specifically, under a random error assumption with the same error strength, the error of the protocol, which is expressed as a distance of measurement,  is in a higher-order formation as with respect to individual measurement devices. 

The above analysis is under the assumption that basis states have a discrepancy $\epsilon$. However, we didn't consider the size of the device. The more qubits involved, the basis states would be less accurate. As $\mathcal{O}(poly(n))$ quantum gates are required to implement specific measurement operators. With a reasonable assumption that each element gate is with a bound error $\varepsilon$, $\epsilon\sim \mathcal{O}(poly(n))\varepsilon)$. Accordingly, the total error caused by measurement setups of our protocol is in polynomials with respect to each individual quantum gate and the size of the system.

\end{document}